\documentclass[a4paper,UKenglish,cleveref, autoref, thm-restate]{lipics-v2021}
\usepackage{amsmath}
\usepackage{nccmath}
\usepackage{float}

\usepackage{amsmath}
\usepackage{color}
\usepackage{pgfplots, pgfplotstable}
\usepackage[colorinlistoftodos]{todonotes} %todo

\newcommand{\xmean}{\bar{x}}
\newcommand{\valueset}{C_M}
\newcommand{\valuescheme}{T_M}
\newcommand{\Iupper}{\Bar{I}}
\newcommand{\Iuppermean}{\Tilde{\mu}_{\Iupper}}
\newcommand{\Iuppervar}{\Tilde{\sigma}_{\Iupper}}
\newcommand{\samplemean}{\Hat{\mu}_{\Iupper}}
\newcommand{\samplevar}{\Hat{\sigma}_{\Iupper}}

\title{Probing the Information Theoretical Roots of Spatial Dependence Measures} %TODO Please add

\author{Zhangyu Wang}{University of California Santa Barbara, Santa Barbara, CA, United States}{zhangyuwang@ucsb.edu}{https://orcid.org/0009-0004-4728-4458}{This work was supported by the National Science Foundation under Grant No. 2033521 A1 – KnowWhereGraph: Enriching and Linking Cross-Domain Knowledge Graphs using Spatially-Explicit AI Technologies.}%TODO mandatory, please use full name; only 1 author per \author macro; first two parameters are mandatory, other parameters can be empty. Please provide at least the name of the affiliation and the country. The full address is optional. Use additional curly braces to indicate the correct name splitting when the last name consists of multiple name parts.

\author{Krzysztof Janowicz}{Faculty of Geosciences, Geography and Astronomy, University of Vienna, Austria \and University of California, Santa Barbara, CA, USA}{krzysztof.janowicz@univie.ac.at}{[orcid]}{}

\author{Gengchen Mai}{SEAI Lab, Department of Geography and the Environment, University of Texas at Austin, TX, USA \and Department of Geography, University of Georgia, GA, USA}{gengchen.mai@austin.utexas.edu}{https://orcid.org/0000-0002-7818-7309}{}

\author{Ivan Majic}{University of Vienna, Austria}{ivan.majic@univie.ac.at}{https://orcid.org/0000-0002-0834-3791}{}

\authorrunning{Z. Wang, K. Janowicz, G. Mai and I. Majic} %TODO mandatory. First: Use abbreviated first/middle names. Second (only in severe cases): Use first author plus 'et al.'
\Copyright{Zhangyu Wang, Krzysztof Janowicz, Gengchen Mai and Ivan Majic} %TODO mandatory, please use full first names. LIPIcs license is "CC-BY";  http://creativecommons.org/licenses/by/3.0/

 %TODO mandatory; please add comma-separated list of keywords

\category{} %optional, e.g. invited paper

\relatedversion{There is an ArXiv version of the full paper: } 
\relatedversiondetails[cite={wang2024probing}]{Full Version}{https://arxiv.org/abs/2405.18459}

\supplement {The source code for computing spatial self-information and reproducing the experiment results in our paper can be found here: }
\supplementdetails [subcategory ={Source Code}]{Software}{https://github.com/Octopolugal/Spatial-Self-Information}

\EventEditors{Benjamin Adams, Amy Griffin, Simon Scheider, and Grant McKenzie}
\EventNoEds{4}
\EventLongTitle{16th International Conference on Spatial Information Theory (COSIT 2024)}
\EventShortTitle{COSIT 2024}
\EventAcronym{COSIT}
\EventYear{2024}
\EventDate{September 17--20, 2024}
\EventLocation{Qu\'{e}bec City, Canada}
\EventLogo{}
\SeriesVolume{315}
\ArticleNo{9}
\begin{document}
\nolinenumbers
\sloppy

\maketitle

\begin{abstract}Intuitively, there is a relation between measures of spatial dependence and information theoretical measures of entropy. For instance, we can provide an intuition of why spatial data is special by stating that, on average, spatial data samples contain less than expected information. Similarly, spatial data, e.g., remotely sensed imagery, that is easy to compress is also likely to show significant spatial autocorrelation. Formulating our (highly specific) core concepts of spatial information theory in the widely used language of information theory opens new perspectives on their differences and similarities and also fosters cross-disciplinary collaboration, e.g., with the broader AI/ML communities. Interestingly, however, this intuitive relation is challenging to formalize and generalize, leading prior work to rely mostly on experimental results, e.g., for describing landscape patterns. In this work, we will explore the information theoretical roots of spatial autocorrelation, more specifically Moran's I, through the lens of self-information (also known as surprisal) and provide both formal proofs and experiments.
\end{abstract}

\section{Introduction}
To explain why \textit{spatial is special} we often list the characteristics underlying spatial data and the processes that created them by pointing to classics such as the modifiable areal unit problem (MAUP) \cite{fotheringham1991modifiable}, spatial dependence and interaction, scale, edge effects, and so forth. However, there are many alternative formulations, some of which draw a more direct relation to our neighboring academic disciplines. For instance, when explaining Tobler's %famous 
First Law of Geography \cite{tobler1970computer,tobler2004first} to our colleagues in the broader AI/ML community by introducing terms such as spatial dependence and spatial autocorrelation measures such as Moran's I, we could instead state that \textit{spatial is special because on average a sample of spatial data contains less than expected information.}

While highlighting different aspects, e.g., omitting the explicit \textit{nearness} from the original definition, such an information theoretic perspective is in many ways equivalent, yet opens up different avenues for understanding why spatial (data) is special from the viewpoint of neighboring but substantially larger disciplines without the need to introduce our own terminology. The information-theoretic statement above also gives rise to an entropy-based understanding of spatial dependence that translates seamlessly into common loss functions in machine learning, such as cross-entropy for classification.

Similar observations be them in space, time, or spacetime can be made about information compression. Intuitively, data, e.g., remotely sensed imagery, with a high degree of spatial autocorrelation should be easier to compress than data with almost no spatial dependence. What is true for simply run-length compression or quad-trees also holds for more abstract situations. For instance, according to lossless compression techniques such as Huffman coding \cite{moffat2019huffman}, daily weather reports (sunny, cloudy, rainy, ...) for Santa Barbara can utilize fewer bits of information than for Vienna. Conversely, very compressible data is also more likely to show strong spatial autocorrelation.

These thoughts do not imply that measures of information compression or entropy can (or should) replace our domain-specific measures such as Moran's I, Geary’s C, semivariograms, and so forth, but that it is worth exploring their commonalities (and differences). For instance, while high (Shannon) entropy implies low spatial autocorrelation, measures such as Moran's I also have an explicit notion of neighborhood encoded via their weights matrix. Hence, a binary checkboard pattern would yield $I= -1$ while a binary partition in two areas would approximate $I \approx 1$. From an entropy perspective, these two patterns are similar as the proportion of blacks and whites, e.g., cells, remains the same. Similar arguments can be made for the compression examples. However, adding the required neighborhood notion to (discrete) entropy is possible, as will be shown below for self-information (i.e., \textit{surprisal}) for the case of classed data, be they rasters or vectors. 

Interestingly, while the relationship between information theory and (spatial) autocorrelation has been noted by others before, the formalization is surprisingly challenging, leading most prior work to take a largely experimental stance. In our work here, we will provide both an experimental intuition and more formal proofs for the proposed \textit{spatial self-information}. Summing up, exploring the information-theoretical roots of spatial dependence and, more specifically, Moran's I is worthwhile for at least the \textbf{following reasons}:

\begin{enumerate}
    \item \textbf{Fostering cross-disciplinarity}: Our community has shown that spatially explicit machine learning models do not only increase the accuracy of (Geo)AI models when applied to geographic data but also inform and improve more general models in various domains \cite{mac2019presence,mai2020multi,mai2022review,mai2023sphere2vec,cole2023spatial,vivanco2023geoclip,russwurm2023geographic}, e.g., leading to a broad interest in location encoding methods outside of GeoAI. Conversely, researchers from the broader AI community \cite{mehrabi2021survey} try to utilize notions such as the MAUP to study problematic coverage and representation biases in training data for image-based foundation models. Yet collaboration and reuse of prior results are sometimes hindered by our highly specific terminology and methods. Casting spatial core concepts in the shared language of information theory may mitigate these issues and also accelerate progress.
    \item \textbf{Quantifying spatial patterns:} Incorporating results from information-theoretical and physical entropy (and related ideas) may open up new avenues to describe complex spatial patterns beyond what is currently available in our spatial analytics toolbox as recently demonstrated by the use of configurational entropy for complex landscapes \cite{cushmanCalculationConfigurationalEntropy2018}.
    \item \textbf{Spatial Data Science education:} Most introductory textbooks on GIS, GIScience, or Spatial Data Science barely make a connection between spatial dependence, information (e.g., image) compression, information theory, and so forth while covering all of them to at least some extent.\footnote{O'Sullivan's and Unwin's \textit{Geographic Information Analysis} being a rare exception.} This makes it difficult for students to grasp the bigger picture, e.g., when meeting entropy again while studying spatial clustering and classification.
\end{enumerate}

\section{Motivation and Related Works} \label{sec:motivation}

Spatial autocorrelation has long been a research focus for both GIScience \cite{goodchild1986spatial,Getis2010} and statistics \cite{cliff1975model}. Efforts have been made to develop statistics that test whether a sample of spatially distributed data is autocorrelated, i.e., against the null hypothesis that the spatial arrangement of the data is randomly generated. Commonly used statistics include Moran's I \cite{moranNotesContinuousStochastic1950, wrigley1982spatial}, Geary's C \cite{geary1954contiguity}, and so forth. Apart from autocorrelation statistics, information-theory-based measures like Batty's spatial entropy~\cite{battySpatialEntropy1974} and S statistics~\cite{ceccatoNewInformationTheoretical}, which extends Moran's I by assuming the observed values are probabilities, have also been studied. However, how to relate these two types of measures lacks in-depth investigation.  

As discussed in the introduction, we are motivated to connect spatial autocorrelation statistics with information-theoretic quantities like entropy and self-information, for the sake of relating theoretical concepts from different disciplines and exploring wider applications (e.g., introducing spatial autocorrelation in loss functions). More specifically, we wish to quantify the self-information, i.e., the \textit{surprisal}, of observing a sample with a certain degree of spatial autocorrelation. The intuition is that higher spatial autocorrelation implies more regular spatial patterns, which is more surprising.

Unfortunately, research shows that there is no general relation between an autocorrelation statistic and its corresponding self-information \cite{chapeau2007autocorrelation}. The information-theoretic counterpart of a spatial autocorrelation statistic needs to be established case by case. In this paper, we aim at deriving that of the (global) Moran's I.

In general, we need to know the probability of observing a certain type of sample to obtain its surprisal. In our case, this means knowing the probability of observing a sample with a certain value of Moran's I. This is an under-studied topic due to the difficulty of deriving the analytical distribution of Moran's I. Instead, researchers use permutation inference to empirically compute the reference distribution. This is good enough for hypothesis tests, as we only need the $p$-values, but not enough for computing the self-information.

Attempts have been made to study the asymptotic behavior of Moran's I under the assumption of knowing the specific underlying stochastic process. For example, Kelejian et al \cite{HKELEJIAN2001219} derived the analytical distribution of Moran's I by assuming the spatial data are generated by a linear model, and the regression disturbance is a known priori or estimated from data. This assumption is reasonable in some areas such as economics, but not necessarily appropriate in geology, urban planning, landscape, remote sensing, etc. 

In many cases, the underlying stochastic process is unknown, and all we can rely on is a broad assumption of randomness. In this sense, it is a combinatorics problem with a strong relation to entropy. Some researches approach a simplified version of this problem via Shannon entropy of co-occurrence counts \cite{leibovici2009defining, nowosad2019information}. They only consider categorical differences, i.e., whether neighboring observations are of the same class, without addressing the numerical differences that are present in classic spatial autocorrelation statistics. Cushman \cite{cushmanCalculationConfigurationalEntropy2018} took an important step in incorporating numerical differences by empirically revealing that the distribution of the total weighted distance between different-valued cells on a grid approximates the distribution of Moran's I. We are inspired by this observation and have decided to provide a more formal analytical analysis. 

In short, we prove why Moran's I asymptotically follows a normal distribution, and its analytical form can be specified without sampling and estimation. This allows efficient and accurate calculation of the probability, consequently the self-information, of any spatial sample with a binary weight matrix. Without the assumptions of the underlying stochastic process, this self-information of spatial autocorrelation can be applied to a wider range of fields and is computationally friendly to be combined with learning algorithms.

\section{Method}

In this section we provide a comprehensive analysis of the asymptotic distribution of Moran's I with binary weights. First we formally define the problem we want to solve and necessary notations in Section \ref{sec:problem-setup}, elaborate the proof from Section \ref{sec:rearrange} to Section \ref{sec:analytical}, and finally reach the core results of this research, i.e., the analytical approximation of Moran's I distribution specified in Theorem \ref{theorem:approx-mean-theorem} and Theorem \ref{theorem:approx-variance-theorem}.

\subsection{Problem Setup} \label{sec:problem-setup}

Consider a sample of $N$ indexed observations $\{x_1, x_2, \cdots x_{N}\}$. We assume the observations take limited discrete \textit{values}, i.e., $\forall i, x_i \in \valueset := \{ c_1, c_2, \cdots c_M \} \text{, where } M \ll N$. Next, we define \textit{value size} as $n_{p} := |\{x_i|x_i = c_p\}|$, i.e., the number of observations in the sample whose values are $c_p$. We call the set $\valuescheme := \{(c_1, n_1), (c_2, n_2), \cdots (c_M, n_M)\}$ the \textit{value scheme} of a sample, i.e., the discrete values and their numbers of occurrences in the sample. The value scheme measures the intrinsic variance in the sample itself regardless of the spatial arrangement. It uniquely determines the sample mean $\xmean = \sum_{i=1}^N x_i / N = \sum_{p=1}^M c_p n_p / \sum_{p=1}^M n_p$ and the sample variance $\sum_{i=1}^N (x_i - \xmean)^2 / N = \sum_{p=1}^M (c_p - \xmean)^2 n_p / \sum_{p=1}^M n_p$.

% {\color{red} histogram}

The binary spatial weight is defined as $w_{i,j} = \mathbb{I}\{x_i~\text{and}~x_j~\text{are neighbors}\}$. Then we can form a directed graph $G$ with $N$ vertices $V = \{v_1, v_2, \cdots v_N\}$ where $v_i$ corresponds to the $i$-th observation $x_i$, and $(v_i, v_j)$ is an edge if and only if $w_{i,j} = 1$. We further require that the degree of each vertex is a fixed number $k \ll N$ (e.g., for the conventional rook and queen weights, $k = 4$ and $k = 8$, respectively, ignoring the border and corner variables). 

Recall that Moran's I is defined as:

$$I = \dfrac{N}{\sum_{i=1}^N\sum_{j=1}^N w_{i,j}} \dfrac{\sum_{i=1}^N\sum_{j=1}^N w_{i,j} (x_i - \xmean) (x_j - \xmean) }{ \sum_{i=1}^N (x_i - \xmean)^2 }$$ 

Moran's I values of two samples are not directly comparable when their value schemes are vastly different. For example, a grid with randomly arranged high-variance values and the same grid with checkerboard arranged low-variance values both have Moran's I $\approx 0$. However, the latter case obviously demonstrates higher spatial autocorrelation and should be differentiated from the former. Therefore, throughout this paper, $\mathbb{P}(I=\alpha)$ (the probability of Moran's I being $\alpha$) always refers to the conditional probability $\mathbb{P}(I=\alpha | \valuescheme)$ (the probability of Moran's I being $\alpha$ given a known and fixed value scheme).

When $\valuescheme$ is known and fixed, $N = \sum_{p=1}^{M} n_p$, $\sum_{i=1}^N\sum_{j=1}^N w_{i,j} = kN$, $\xmean$ and $\sum_{i=1}^N (x_i - \xmean)^2$, are also known and fixed. Define $\Iupper := \sum_{i=1}^N\sum_{j=1}^N w_{i,j} (x_i - \xmean) (x_j - \xmean)$ as the \textit{unscaled Moran's I}. Then $I \propto \Iupper$. For the simplicity of discussion, in the following, we will focus on $\Iupper$.

Our purpose is to find the distribution of $\Iupper$ (and naturally that of $I$) so that for any $\alpha \in \mathbb{R}$ we know $\mathbb{P}(\Iupper = \alpha)$. Then we define \textit{spatial self-information} $J := - \log (\mathbb{P}(\Iupper = \alpha))$, which quantifies the surprisal of observing a sample with a certain degree of spatial autocorrelation.

To achieve this, we will prove that the distribution of $\Iupper$ can be asymptotically approximated by a sum of normal distributions in the rest of the section. The proof is structured as follows:

\begin{enumerate}
    \item First, in Section \ref{sec:rearrange}, we rearrange $\Iupper$ as a weighted sum of several random variables;
    \item Then, in Section \ref{sec:asymptotic-binary}, we prove that the distributions of these random variables asymptotically follow binomial/Poisson binomial distributions under certain necessary assumptions;
    \item Finally, in Section \ref{sec:normal-approx}, we show that these distributions can be approximated by normal distributions, which leads to the conclusion that $\Iupper$ also follows a normal distribution. We further derive the analytical form of this distribution in Section \ref{sec:analytical}.
\end{enumerate}

Following that, we develop a set of techniques in Section \ref{sec:accuracy-robustness} that correct the error caused by violations of assumptions and conditions made during the problem set-up and the proof. Experiments demonstrate that with these corrections our approximation is robust to relaxations. It enables practical application in real-world, i.e., non-ideal situations.

\subsection{Rearrangement of $\Iupper$} \label{sec:rearrange}

\begin{definition}
A tuple $(x_i, x_j)$ is called a $pq$-pair if $x_i=c_p, x_j=c_q, w_{i,j} = 1$. We say a $pq$-pair starts with $c_p$ and ends with $c_q$.
\end{definition}

\begin{definition}
$S_{p,q} := \{(x_i,x_j) | x_i = c_p,~x_j = c_q~\text{and}~w_{i,j} = 1\}$, $c_p, c_q \in \valueset$ is the set of pq-pairs. If $p \neq q$, we say $S_{p,q}$ is a different-value set; otherwise a same-value set.
\end{definition}

It is worth noting that the order of indices matter. In a same-value set $S_{p,p}$, $(x_i,x_j)$ and $(x_j, x_i)$ are counted as two different $pp$-pairs even though $x_i = x_j = c_p$.

\begin{lemma}[Rearrangement of $\Iupper$ as Cardinality of Sets]
\label{lemma:rearrangement}
$\Iupper$ is a weighted sum of the cardinality of all possible sets of $pq$-pairs. Specifically, 
$\Iupper = \sum_{p,q} (c_p - \xmean)(c_q - \xmean) |S_{p,q}|$.
\end{lemma}

\begin{proof}

$$
\begin{aligned}
& \sum_{i=1}^N\sum_{j=1}^N w_{i,j} (x_i - \xmean) (x_j - \xmean) \\
& = \sum_{p,q} \left[\sum_{i=1}^N\sum_{j=1}^N \mathbb{I}\{x_i=c_p, x_j=c_q\} w_{i,j} (x_i - \xmean) (x_j - \xmean)\right] \\
% & \text{split the summation into groups by the values of x_i and x_j} \\
 & = \sum_{p,q} \left[\sum_{i=1}^N\sum_{j=1}^N \mathbb{I}\{x_i=c_p, x_j=c_q, w_{i,j}=1\} (x_i - \xmean) (x_j - \xmean)\right] \\
 & = \sum_{p,q} \left[\sum_{i=1}^N\sum_{j=1}^N \mathbb{I}\{x_i=c_p, x_j=c_q, w_{i,j}=1\} (c_p - \xmean) (c_q - \xmean)\right] \\
 & = \sum_{p,q} (c_p - \xmean) (c_q - \xmean) \left[\sum_{i=1}^N\sum_{j=1}^N \mathbb{I}\{x_i=c_p, x_j=c_q, w_{i,j}=1\}\right] \\
 % & \text{given $x_i=c_p$ and $x_j=c_q$, (x_i - \xmean) (x_j - \xmean) is a constant number.} \\
 & = \sum_{p,q} (c_p - \xmean) (c_q - \xmean) |S_{p,q}| \\
 % & \text{$\sum_i\sum_j \mathbb{I}\{X_i=c_p, X_j=c_q, w_{i,j}=1\}$ exactly counts the number of all $pq$ edges}
\end{aligned}
$$

% $\sum_i\sum_j w_{i,j} (X_i - \xmean) (X_j - \xmean) = \sum_{p=1}^{M} \sum_{q=1}^{M} [\sum_i\sum_j \mathbb{I}\{X_i=c_p, X_j=c_q\} w_{i,j} (X_i - \xmean) (X_j - \xmean)]$. By noticing that for any $pq$-edge, $(X_i - \xmean) (X_j - \xmean) = (c_p - \xmean)(c_q - \xmean)$ is a fixed number, $\sum_{p=1}^{M} \sum_{q=1}^{M} [\sum_i\sum_j \mathbb{I}\{X_i=c_p, X_j=c_q\} w_{i,j} (X_i - \xmean) (X_j - \xmean)] = \sum_{p=1}^{M} \sum_{q=1}^{M} [\sum_i\sum_j \mathbb{I}\{X_i=c_p, X_j=c_q, w_{i,j}=1\} (X_i - \xmean) (X_j - \xmean)] = \sum_{p=1}^{M} \sum_{q=1}^{M} (X_i - \xmean) (X_j - \xmean) [\sum_i\sum_j \mathbb{I}\{X_i=c_p, X_j=c_q, w_{i,j}=1\}]$. The summation $\sum_i\sum_j \mathbb{I}\{X_i=c_p, X_j=c_q, w_{i,j}=1\}$ is exactly $|S_{p,q}|$ because it counts the number of all $pq$ edges.

% \todo{A figure here to show what "edges" mean}

\end{proof}

Now, if we know the distributions of each $|S_{p,q}|$, we can analyze the distribution of their weighted sum. We will prove that for different-value sets, $|S_{p,q}|$ asymptotically follows a binomial distribution, and for same-value sets, $|S_{p,q}|$ asymptotically follows a Poisson binomial distribution. 
% Most importantly, both cases can be approximated by normal distributions. Then by the fact that normal distributions are stable, the weighted sum of $|S_{p,q}|$ also follows a normal distribution.

\subsection{Asymptotic Binomial and Poisson Binomial Distributions} \label{sec:asymptotic-binary}

Let $B(p, n)$ denote the binomial distribution with probability of success $p$ and number of trials $n$, $PB(p_i; i=1,2,\cdots,n)$ denote the Poisson binomial distribution with probabilities of success $p_1, p_2, \cdots, p_n$, and $N(\mu, \sigma^2)$ denote the normal distribution with mean $\mu$ and variance $\sigma^2$.

\begin{lemma}[Probability of the Cardinality of Same-Value Sets]
\label{lemma:same-value-lemma}
If $kn_p \ll N$, then $\dfrac{1}{2}|S_{p,p}| \sim PB(\dfrac{k(t-1)}{N}; t = 1,\cdots, n_p)$,
with mean $\mu_{p,p} = \dfrac{1}{2}(n_p - 1)\dfrac{kn_p}{N}$ and variance $\sigma^2_{p,p} = \dfrac{1}{2}(n_p-1)\dfrac{kn_p}{N}\left[1 - \dfrac{k(2n_p-1)}{3N}\right]$.
\end{lemma}

\begin{proof}

Let black represents $c_p$. Consider a one-phase graph coloring process. $G$ is the uncolored graph. We randomly (with equal probability) color a vertex black for $n_p$ times, resulting in a black-colored graph $G_{b}$. The probability of having $l$ edges of same colored vertices in $G_{b}$ equals the probability of $|S_{p,p}| = 2l$ because each edge of the same colored vertices will be counted twice in $|S_{p,p}|$.

To obtain the general form of the asymptotic distribution of $|S_{p,p}|$, we need an assumption that none of the colored pairs share common neighbors except themselves, because the probability of two vertices sharing common neighbors depends on the specific structure of the uncolored graph $G$. Say coloring a neighbor of a black vertex to be a $success$, and otherwise a $failure$. This assumption guarantees that 1) each success creates and only creates two $pp$-pairs, 2) each success adds $k-2$ edges that have and only have one vertex colored black and each failure adds $k$ such edges, and 3) the probability of coloring a neighbor of a black vertex is independent of the previous coloring result. 

The first time of coloring will not result in any edges of the same colored vertices. For the second time, the probability of coloring a neighbor of a black vertex, i.e., success in creating two $pp$-pairs, is $\dfrac{k}{N-1} \approx \dfrac{k}{N}$. Then for the third time, the probability of coloring a neighbor of a black vertex becomes $\dfrac{k}{N-1}\dfrac{2k-2}{N-2} + \left(1 - \dfrac{k}{N-1}\right)\dfrac{2k}{N-2} = \dfrac{2k^2}{(N-1)(N-2)} - \dfrac{2k}{(N-1)(N-2)} + \dfrac{2k}{N-2} - \dfrac{2k^2}{(N-1)(N-2)} = \dfrac{2k(N-2)}{(N-1)(N-2)} = \dfrac{2k}{N-1} \approx \dfrac{2k}{N}$. It is easy to verify that for the $t$-th time, the probability of success, i.e., creating two $pp$-pairs, is approximately $\dfrac{(t-1)k}{N}$. 

We can view the coloring process as a series of $n_p$ independent binary trials with increasing probabilities of success. Then $\dfrac{1}{2}|S_{p,p}| = l$ equals the total number of success among all $n_p$ trials, which follows a Poisson binomial distribution $PB(\dfrac{k(t-1)}{N}; t = 1,\cdots, n_p)$, $\mu_{p,p} = \sum_{t=1}^{n_p} \dfrac{(t-1)k}{N} = \dfrac{1}{2}(n_p - 1)\dfrac{kn_p}{N}, \sigma^2_{p,p} = \sum_{t=1}^{n_p} \left[\dfrac{(t-1)k}{N}\left(1-\dfrac{(t-1)k}{N}\right)\right] = \sum_{t=1}^{n_p} \dfrac{(t-1)k}{N} - \sum_{t=1}^{n_p} \left(\dfrac{(t-1)k}{N}\right)^2 = \dfrac{1}{2}(n_p-1)\dfrac{kn_p}{N}\left[1 - \dfrac{k(2n_p-1)}{3N}\right]$ asymptotically as $\dfrac{n_p}{N} \rightarrow 0$. 

\end{proof}

\begin{lemma}[Probability of the Cardinality of Different-Value Sets]
\label{lemma:different-value-lemma}
If $n_q \ll kn_p \ll N$, then $|S_{p,q}| \sim B(n_q, \dfrac{kn_p}{N})$, with mean
$\mu_{p,q} = n_q\dfrac{kn_p}{N}$ and variance $\sigma^2_{p,q} = n_q \dfrac{kn_p}{N} \left(1 -\dfrac{kn_p}{N}\right)$.
\end{lemma}

\begin{proof}
The proof uses similar techniques as in Lemma \ref{lemma:same-value-lemma}. Let black represents $c_p$ and white represents $c_q$. Consider a two-phase graph coloring process. $G$ is the uncolored graph. In the first phase, randomly (with equal probability) color $n_p$ vertices black, resulting in a black-colored graph $G_b$; in the second phase, randomly (with equal probability) color $n_q$ vertices white, resulting in a black-white-colored graph $G_{bw}$. The probability of having $l$ edges of differently colored vertices in $G_{bw}$ equals the probability of $|S_{p,q}| = l$.

Since $kn_p \ll N$, the probability that we get a $G_b$ in which two black vertices have common neighbors is very small. Thus the assumption we mentioned in the proof of Lemma \ref{lemma:same-value-lemma} can be considered valid, i.e., we can assume that in all possible outcomes of $G_b$, all black vertices do not share neighbors, which means in $G_b$ there are in total $kn_p$ edges that have and only have one vertex colored black. 

Then we randomly color $n_q$ vertices white based on $G_b$. Consider this process as repeatedly coloring one random vertex white for $n_q$ times. For the first time, the probability of coloring a neighbor of a black vertex, i.e., success in creating one $pq$-pair, is $\dfrac{kn_p}{N - n_p} \approx \dfrac{kn_p}{N}$. For the second time, the probability becomes $\dfrac{kn_p}{N - n_p} \dfrac{kn_p - 1}{N - n_p - 1} + \left(1 - \dfrac{kn_p}{N - n_p}\right)\dfrac{kn_p}{N - n_p - 1} \approx \dfrac{kn_p}{N} \dfrac{kn_p - 1}{N} + \left(1 - \dfrac{kn_p}{N}\right)\dfrac{kn_p}{N} \approx \dfrac{kn_p}{N}$. That is, when $kn_p \ll N$, the second step of coloring is approximately independent of the first step with the same probability of success. It can be easily verified that this approximate independence and equal probability holds for all $n_q$ steps as long as $n_q \ll kn_p$. The intuition is simple: if you draw a dozen out of thousands of black and white balls, the color of each draw is almost independent of other draws with equal possibilities. Subsequently, the second phase of coloring can be viewed approximately as $n_q$ i.i.d. binary trials with probability of success $\dfrac{kn_p}{N}$. Each successful trial adds a black-white edge to $G_{bw}$. 
We immediately know for any $G_b$, $|S_{p,q}| \sim B(n_q, \dfrac{kn_p}{N})$, $\mu_{p,q} = n_q\dfrac{kn_p}{N}, \sigma_{p,q}^2 = \dfrac{kn_p n_q}{N} (1 -\dfrac{kn_q}{N})$ asymptotically as $\dfrac{n_p}{N} \rightarrow 0$ and $\dfrac{n_q}{N} \rightarrow 0$ because the sum of i.i.d. binary random variables is a binomial random variable.

\end{proof}

\subsection{Normal Approximation of Binomial and Poisson Binomial Distributions} \label{sec:normal-approx}

Whereas we have demonstrated that $\Iupper$ is a weighted sum of binomial and Poisson binomial random variables, the analytical form of its distribution can not be simply found. Instead, if we approximate the binomial and Poisson binomial distributions with normal distributions, by the fact that normal distributions are stable distributions, the weighted sum can also be approximated by a normal distribution.

% By the De Moivre-Laplace theorem~\cite{papoulis2002probability}, a binomial distribution with a relatively large probability of success (e.g., $> 0.1$) can be well approximated by a normal distribution with the same mean and variance. And in \cite{approximate-poisson-binomial}, the authors demonstrate that a Poisson binomial distribution consisting of $n$ independent binary trials with mean $\mu$ and variance $\sigma^2$ can be approximated by a normal distribution $N(\mu - \dfrac{1}{2}, \sigma^2)$ when $n$ is sufficiently large. Additionally, for a normal random variable $Y \sim N(\mu, \sigma^2)$ and any constants $a, b~(a \neq 0)$, $aY+b \sim N(a\mu+b, a^2Y)$. 

% Therefore, we have the following two normal approximation lemmas:

% Therefore, for different-value sets, $|S_{p,q}| \sim N(n_q\dfrac{kn_p}{N}, n_q \dfrac{kn_p}{N} \left(1 -\dfrac{kn_p}{N}\right))$, and for same-value sets, $\dfrac{1}{2}|S_{p,p}| \sim N(\dfrac{1}{2}(n_p - 1)\dfrac{kn_p}{N} - \dfrac{1}{2}, \dfrac{1}{2}(n_p-1)\dfrac{kn_p}{N}\left[1 - \dfrac{k(2n_p-1)}{3N}\right])$, i.e., $|S_{p,p}| \sim N((n_p - 1)\dfrac{kn_p}{N} - 1, 2(n_p-1)\dfrac{kn_p}{N}\left[1 - \dfrac{k(2n_p-1)}{3N}\right])$.

By the De Moivre-Laplace theorem~\cite{papoulis2002probability}, a binomial distribution with a relatively large probability of success (e.g., $> 0.1$) can be well approximated by a normal distribution with the same mean and variance. Therefore, we have: 
\begin{lemma}[Normal Approximation for Different-Value Sets]
\label{lemma:normal-approx-different-lemma}
If $n_q \ll kn_p \ll N$, then $|S_{p,q}| \sim N\Big(n_q\dfrac{kn_p}{N}, n_q \dfrac{kn_p}{N} \left(1 -\dfrac{kn_p}{N}\right)\Big)$ approximately.
\end{lemma}

In addition, in \cite{approximate-poisson-binomial}, the authors demonstrate that a Poisson binomial distribution consisting of $n$ independent binary trials with mean $\mu$ and variance $\sigma^2$ can be approximated by a normal distribution $N(\mu - \dfrac{1}{2}, \sigma^2)$ when $n$ is sufficiently large. According to \cite{approximate-poisson-binomial}, we have:
\begin{lemma}[Normal Approximation for Same-Value Sets]
\label{lemma:normal-approx-same-lemma}
If $kn_p \ll N$ and $n_p$ is sufficiently large, then $|S_{p,p}| \sim N\Big((n_p - 1)\dfrac{kn_p}{N} - 1, 2(n_p-1)\dfrac{kn_p}{N}\left[1 - \dfrac{k(2n_p-1)}{3N}\right]\Big)$ approximately.
\end{lemma}

According to \cref{lemma:normal-approx-different-lemma} and \cref{lemma:normal-approx-same-lemma}, both $|S_{p,q}|$ and $|S_{p,p}|$ can be approximated as normal distributions. Additionally, we know that for a normal random variable $Y \sim N(\mu, \sigma^2)$ and any constants $a_1, a_2, a_1 \neq 0$, $a_1Y+a_2 \sim N(a_1\mu+a_2, a_1^2\sigma^2)$. Since \cref{lemma:rearrangement} shows that $\Iupper$ is a weighted sum of $|S_{p,q}|$ and $|S_{p,p}|$, we can derive that $\Iupper$ can be also approximated with a normal distribution.

\subsection{Analytical Approximation of the Distribution of $\Iupper$} \label{sec:analytical}

We know from the discussions above that $\Iupper$ approximately follows a normal distribution. The parameters we need to specify are its mean and variance. 

\begin{theorem}[Approximate Mean of $\Iupper$]
\label{theorem:approx-mean-theorem}
Given $\valuescheme$, the approximate mean of $\Iupper$ is 
\begin{ceqn}
\begin{align} \label{eq:mean-of-Iupper}
\Iuppermean = \sum_{p \neq q} (c_p - \xmean)(c_q - \xmean) \mu_{p,q} + \sum_{p} (c_p - \xmean)^2 \mu_{p,p}
\end{align}
\end{ceqn}
where $\mu_{p,q} = \min(n_p, n_q)\dfrac{k\max(n_p, n_q)}{N}$, $\mu_{p,p} = (n_p - 1)\dfrac{kn_p}{N} - 1$
\end{theorem}

\begin{proof}
This is a simple consequence of the fact that for normal random variables, the mean of the sum equals the sum of the means. The minimum and the maximum functions are included to satisfy the condition that $n_q \ll kn_p$ in Lemma \ref{lemma:different-value-lemma}.
\end{proof}

It is more complicated to derive the variance. We know for independent normal random variables, the variance of the sum equals the sum of variances. However, this independence requirement does not hold for all $|S_{p,q}|$. When the spatial arrangement of $M-1$ values is known, the spatial arrangement of the remaining value is automatically known. That is, for any $1 \leq r \leq M$, $\sum_{p \neq r, q \neq r} |S_{p,q}|$ and $\sum_{p = r~\text{or}~q = r} |S_{p,q}|$ are correlated. Call $c_r$ the \textit{background value} and the other values \textit{foreground values}. The following theorem states that if there is a background value that a sufficiently large proportion of samples takes, we can analytically approximate the variance of $\Iupper$. 

\begin{theorem}[Approximate Variance of $\Iupper$]
\label{theorem:approx-variance-theorem}
Given $\valuescheme$, let $c_{r_{\max}}$ be the value that has the largest value size $n_{r_{\max}}$. If $\dfrac{n_{r_{\max}}}{N}$ is sufficiently large, the approximate variance of $\Iupper$ is 

\begin{equation} \label{eq:variance-of-Iupper}
\begin{split}
% \Iuppervar^2 = & \quad 4 \left \{ \sum_{p \neq r_{\max}, q \neq r_{\max}} (c_p - \xmean)^2(c_q - \xmean)^2\min(n_p, n_q) \dfrac{k\max(n_p, n_q)}{N} \left(1 -\dfrac{k\max(n_p, n_q)}{N}\right) \right. \\
%  & \left. \quad + \sum_{p \neq r_{\max}} 2 (c_p - \xmean)^4 (n_p-1)\dfrac{kn_p}{N}\left[1 - \dfrac{k(2n_p-1)}{3N}\right] \right \} \\
\Iuppervar^2 = & ~ \sum_{p \neq q \neq r_{\max}} \left[ (c_p - \xmean)(c_q - \xmean) - 2 (c_p - \xmean)(c_{r_{\max}} - \xmean) + (c_{r_{\max}} - \xmean)^2 \right]^2  \sigma^2_{p,q}\\
& ~ + \sum_{p \neq r_{\max}} \left[ (c_p - c_{r_{\max}})^2 \right]^2 \sigma^2_{p,p} \\
\end{split}
\end{equation}

where $$\sigma^2_{p,q} = \min(n_p, n_q) \dfrac{k\max(n_p, n_q)}{N} \left(1 -\dfrac{k\max(n_p, n_q)}{N}\right)$$ \\
$$\sigma^2_{p,p} = 2(n_p-1)\dfrac{kn_p}{N}\left[1 - \dfrac{k(2n_p-1)}{3N}\right]$$

\end{theorem}

\begin{proof}
Consider generating a sample by an $M-1$ phase graph coloring process, given $\valuescheme$. Choose $c_r$ as the background value. In each phase, we select an index $i$ from $1$ to $M$ except $r$ without replacement and fill $n_i$ vertices with the color of $c_i$. Regardless of how we select the values, after $M-1$ phases, the way to color the remaining $n_r$ vertices is now fixed --- i.e., the cardinality of the sets related to the last value is fully determined by the cardinality of other sets.

Formally, given any $1 \leq r \leq M$ and $p \neq r$, 

$$kn_p = \sum_{q} |S_{p,q}| = |S_{p,r}| + \sum_{q \neq r} |S_{p,q}|$$

because $\sum_{q} |S_{p,q}|$ is the total number of pairs that start with $c_p$, i.e., the edges in the directed graph $G$ that starts with $v_p$. Then 

\begin{ceqn}
\begin{align} \label{eq:spr}
    |S_{p,r}| = kn_p - \sum_{q \neq r} |S_{p,q}|
\end{align}
\end{ceqn}

By symmetry, $|S_{r, p}| = |S_{p,r}|$. Now consider $|S_{r,r}|$. Similarly, 
$$
\begin{aligned}
kN = & \sum_{p,q} |S_{p,q}| = |S_{r,r}| + \sum_{p \neq r ~\text{or}~ q \neq r} |S_{p,q}| \\
= & ~ |S_{r,r}| + \sum_{q \neq r} |S_{r,q}| + \sum_{p \neq r} |S_{p,r}| + \sum_{p \neq r, q \neq r} |S_{p,q}|
\end{aligned}
$$

Then

\begin{ceqn}
\begin{align} \label{eq:srr}
    |S_{r,r}| = kN - \left(\sum_{q \neq r} |S_{r,q}| + \sum_{p \neq r} |S_{p,r}| + \sum_{p \neq r,q \neq r} |S_{p,q}|\right)
\end{align}
\end{ceqn}

When $\dfrac{n_{r}}{N}$ is sufficiently large, the $M-1$ phases of graph coloring can be considered independent (similar intuition as in the proof of Lemma \ref{lemma:different-value-lemma}). Thus $|S_{p,q}|, p \neq r, q \neq r$, i.e. the cardinality of sets restricted to the foreground values, are approximately independent normal random variables. If we can represent $\Iupper$ as a weighted sum of the cardinality of these sets, the variance of $\Iupper$ can then be represented as a weighted sum of their variances.

\begin{lemma}[Rearrangement of $\Iupper$ as Cardinality of Sets restricted to Foreground Values]
\label{lemma:rearrangement-lemma-foreground}
$\Iupper$ is a weighted sum of the cardinality of all possible sets of $pq$-pairs restricted to the foreground values. Specifically, let $c_r$ be the background value, then

$$\Iupper = Q + \sum_{p \neq r, q \neq r} \left[ (c_p - \xmean)(c_q - \xmean) - 2 (c_p - \xmean)(c_r - \xmean) + (c_r - \xmean)^2 \right] |S_{p,q}|$$
\end{lemma}

where $Q := (c_r - \xmean)^2 kN + 2\sum_{p \neq r} \left[ (c_p - \xmean)(c_r - \xmean) - (c_r - \xmean)^2 \right] kn_p$ is a constant.

\begin{claimproof}

{\allowdisplaybreaks
\begin{ceqn}
\begin{align*}
\Iupper = & \sum_{p,q} (c_p - \xmean) (c_q - \xmean) |S_{p,q}| \\
= & ~ (c_r - \xmean)^2|S_{r,r}| + \sum_{q \neq r} (c_r - \xmean)(c_q - \xmean) |S_{r,q}| + \sum_{p \neq r} (c_p - \xmean)(c_r - \xmean) |S_{p,r}| \\
& ~ + \sum_{p \neq r, q \neq r} (c_p - \xmean)(c_q - \xmean) |S_{p,q}| \\
\end{align*}
\end{ceqn}
}

Use Equation \ref{eq:srr},

{\allowdisplaybreaks
\begin{ceqn}
\begin{align*}
\Iupper = & ~ (c_r - \xmean)^2 \left[ kN - \left(\sum_{q \neq r} |S_{r,q}| + \sum_{p \neq r} |S_{p,r}| + \sum_{p \neq r,q \neq r} |S_{p,q}| \right) \right] \\
& ~ + \sum_{q \neq r} (c_r - \xmean)(c_q - \xmean) |S_{r,q}| + \sum_{p \neq r} (c_p - \xmean)(c_r - \xmean) |S_{p,r}| \\
& ~ + \sum_{p \neq r, q \neq r} (c_p - \xmean)(c_q - \xmean) |S_{p,q}| \\
\end{align*}
\end{ceqn}
}

Use symmetry $|S_{r, p}| = |S_{p,r}|$,

{\allowdisplaybreaks
\begin{ceqn}
\begin{align*}
\Iupper = & ~ (c_r - \xmean)^2 kN - (c_r - \xmean)^2 \left( 2\sum_{p \neq r} |S_{p,r}| + \sum_{p \neq r,q \neq r} |S_{p,q}| \right) \\
& ~ + 2\sum_{p \neq r} (c_p - \xmean)(c_r - \xmean) |S_{p,r}| + \sum_{p \neq r, q \neq r} (c_p - \xmean)(c_q - \xmean) |S_{p,q}| \\
= & ~ (c_r - \xmean)^2 kN + 2\sum_{p \neq r} \left[ (c_p - \xmean)(c_r - \xmean) - (c_r - \xmean)^2 \right] |S_{p,r}| \\
& ~ + \sum_{p \neq r, q \neq r} \left[ (c_p - \xmean)(c_q - \xmean) - (c_r - \xmean)^2 \right] |S_{p,q}| \\
\end{align*}
\end{ceqn}
}

Use Equation \ref{eq:spr},

{\allowdisplaybreaks
\begin{ceqn}
\begin{align*}
\Iupper = & ~ (c_r - \xmean)^2 kN + 2\sum_{p \neq r} \left[ (c_p - \xmean)(c_r - \xmean) - (c_r - \xmean)^2 \right] \left(kn_p - \sum_{q \neq r} |S_{p,q}|\right) \\
& ~ + \sum_{p \neq r, q \neq r} \left[ (c_p - \xmean)(c_q - \xmean) - (c_r - \xmean)^2 \right] |S_{p,q}| \\
\end{align*}
\end{ceqn}
}

Finally, merge like terms with respect to $|S_{p,q}|$,

{\allowdisplaybreaks
\begin{ceqn}
\begin{align*}
\Iupper = & ~ (c_r - \xmean)^2 kN + 2\sum_{p \neq r} \left[ (c_p - \xmean)(c_r - \xmean) - (c_r - \xmean)^2 \right] kn_p \\
& ~ - 2\sum_{p \neq r} \left[ (c_p - \xmean)(c_r - \xmean) - (c_r - \xmean)^2 \right]\sum_{q \neq r} |S_{p,q}| \\
& ~ + \sum_{p \neq r, q \neq r} \left[ (c_p - \xmean)(c_q - \xmean) - (c_r - \xmean)^2 \right] |S_{p,q}| \\
= & ~ (c_r - \xmean)^2 kN + 2\sum_{p \neq r} \left[ (c_p - \xmean)(c_r - \xmean) - (c_r - \xmean)^2 \right] kn_p \\
& ~ - 2\sum_{p \neq r, q \neq r} \left[ (c_p - \xmean)(c_r - \xmean) - (c_r - \xmean)^2 \right] |S_{p,q}| \\
& ~ + \sum_{p \neq r, q \neq r} \left[ (c_p - \xmean)(c_q - \xmean) - (c_r - \xmean)^2 \right] |S_{p,q}| \\
= & ~ (c_r - \xmean)^2 kN + 2\sum_{p \neq r} \left[ (c_p - \xmean)(c_r - \xmean) - (c_r - \xmean)^2 \right] kn_p \\
& ~ + \sum_{p \neq r, q \neq r} \left[ (c_p - \xmean)(c_q - \xmean) - 2 (c_p - \xmean)(c_r - \xmean) + (c_r - \xmean)^2 \right] |S_{p,q}| \\
= & ~ Q + \sum_{p \neq r, q \neq r} \left[ (c_p - \xmean)(c_q - \xmean) - 2 (c_p - \xmean)(c_r - \xmean) + (c_r - \xmean)^2 \right] |S_{p,q}|
\end{align*}
\end{ceqn}
}

\end{claimproof}

As $Q$ remains fixed given $T_M$ and $|S_{p,q}|, p \neq r, q \neq r$ are approximately independent, the following relation holds by linearity:

\begin{align*}
\text{Var}(\Iupper) = & ~ \text{Var} (Q + \sum_{p \neq r, q \neq r} \left[ (c_p - \xmean)(c_q - \xmean) - 2 (c_p - \xmean)(c_r - \xmean) + (c_r - \xmean)^2 \right] |S_{p,q}|) \\
= & ~ \text{Var} (\sum_{p \neq r, q \neq r} \left[ (c_p - \xmean)(c_q - \xmean) - 2 (c_p - \xmean)(c_r - \xmean) + (c_r - \xmean)^2 \right] |S_{p,q}|) \\
\approx & ~ \sum_{p \neq r, q \neq r} \text{Var} (\left[ (c_p - \xmean)(c_q - \xmean) - 2 (c_p - \xmean)(c_r - \xmean) + (c_r - \xmean)^2 \right] |S_{p,q}|)\\
= & ~ \sum_{p \neq r, q \neq r} \left[ (c_p - \xmean)(c_q - \xmean) - 2 (c_p - \xmean)(c_r - \xmean) + (c_r - \xmean)^2 \right]^2 \text{Var}(|S_{p,q}|) \\
= & ~ \sum_{p \neq q \neq r} \left[ (c_p - \xmean)(c_q - \xmean) - 2 (c_p - \xmean)(c_r - \xmean) + (c_r - \xmean)^2 \right]^2 \text{Var}(|S_{p,q}|) \\
& ~ + \sum_{p \neq r} \left[ (c_p - \xmean)^2 - 2 (c_p - \xmean)(c_r - \xmean) + (c_r - \xmean)^2 \right]^2 \text{Var}(|S_{p,p}|) \\
= & ~ \sum_{p \neq q \neq r} \left[ (c_p - \xmean)(c_q - \xmean) - 2 (c_p - \xmean)(c_r - \xmean) + (c_r - \xmean)^2 \right]^2 \text{Var}(|S_{p,q}|) \\
& ~ + \sum_{p \neq r} \left[ (c_p - c_r)^2 \right]^2 \text{Var}(|S_{p,p}|) \\
\end{align*}

Insert the formulae of $\text{Var}(|S_{p,q}|)$ and $\text{Var}(|S_{p,p}|)$ from Lemma \ref{lemma:normal-approx-different-lemma} and Lemma \ref{lemma:normal-approx-same-lemma} and Theorem \ref{theorem:approx-variance-theorem} is proved. To best satisfy the condition of approximate independence, the background value should have the largest value size.
\end{proof}

% \todo{A figure here to illustrate how the sum of variance does not fit the histogram}

\section{Analysis of Approximation Accuracy and Robustness on Synthetic Data} \label{sec:accuracy-robustness}

For the sake of mathematical preciseness, we have introduced many assumptions and conditions during the problem setup and the proof. It is necessary to systematically investigate how violations of these assumptions and conditions may affect the accuracy of approximation. In the following, by using a series of experiments on synthetic data, we demonstrate that most of them can be relaxed while our approximation remains sufficiently accurate. This implies that the spatial self-information $J$ is numerically robust for practical use. 

As a quick recap, the assumptions and conditions used in our theoretical derivations are listed below in the order they are introduced:

\begin{itemize}
    \item \textbf{Assumption 1}: the spatial weights are binary (0-1). See Section \ref{sec:problem-setup}.
    \item \textbf{Assumption 2}: the observations take a small number of discrete values. See Section \ref{sec:problem-setup}.
    \item \textbf{Assumption 3}: the observations have equal numbers of neighbors. See Section \ref{sec:problem-setup}.
    \item \textbf{Assumption 4}: no pairs of same-colored vertices share common neighbors other than themselves. See Section \ref{sec:asymptotic-binary}.
    \item \textbf{Condition 1}: for deriving the asymptotic distributions of $|S_{p,q}|$ and $|S_{p,p}|$, it is required that $n_q \ll kn_p \ll N$. See Section \ref{sec:asymptotic-binary}.
    \item \textbf{Condition 2}: $n_p$ is required to be sufficiently large to have good normal approximations of the asymptotic binomial and Poisson binomial distributions. See Section \ref{sec:normal-approx}.
    \item \textbf{Condition 3}: $n_{r_{\max}} / N$ is required to be sufficiently large to ensure approximate independence of normal random variables. See Section \ref{sec:analytical}.
\end{itemize}

Among all these terms, only Assumption 1 is mandatory. We will analyze how and to what extent the other assumptions and conditions can be relaxed in the following subsections, starting from the most critical ones.

The statistics used to measure the accuracy of approximation are: the analytical mean $\Iuppermean$, the analytical standard deviation $\Iuppervar$, the sample mean $\samplemean$, the sample standard deviation $\samplevar$, the standardized difference of mean $|\Iuppermean - \samplemean| / \Iuppervar$, the standardized difference of standard deviation $|\Iuppervar - \samplevar| / \samplevar$, and the KL-divergence from the analytical normal distribution to the empirical distribution. For all experiments, we randomly sample 10,000 $40\times40$ grids for $10$ times and report both the means and the standard deviations of the statistics.

\subsection{Relaxation of Condition 3: Violation of Approximate Independence}

The approximation accuracy is mostly dependent on the satisfaction of the independence condition in Theorem \ref{theorem:approx-variance-theorem}. The level of independence is measured by $b = n_{r_{\max}} / N$, the proportion of background. The higher $b$, the higher independence. 

\begin{figure*}[t!]
\centering
\captionsetup[subfloat]{justification=centering}
\subfloat[high independence, $b = 65\%$. \\ $\Iuppermean = -5.21$, $\Iuppervar = 135.11$ \\ $\samplemean = -4.51$, $\samplevar = 132.84$]{\includegraphics[width = 0.30\textwidth]{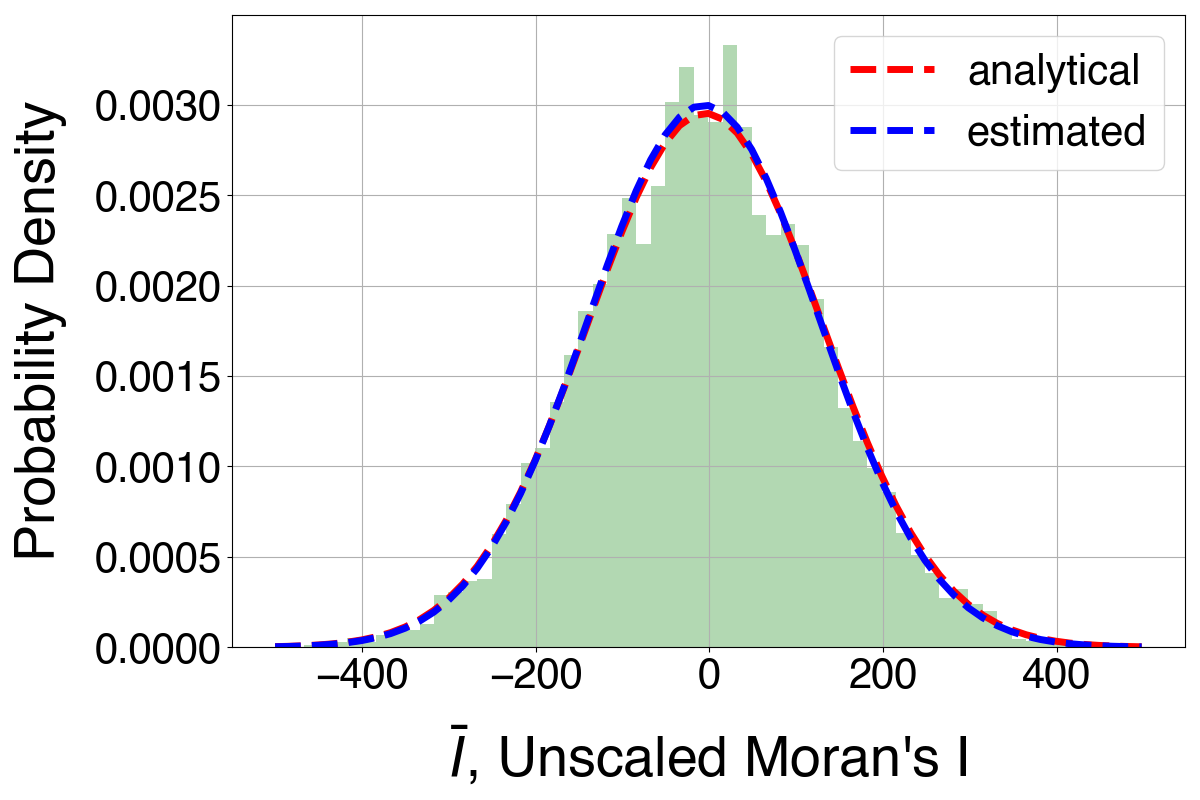}\label{fig:high-independence}}\quad
\subfloat[mid independence, $b = 45\%$. \\ $\Iuppermean = -4.45$, $\Iuppervar = 168.46$ \\ $\samplemean = -5.56$, $\samplevar = 153.99$]{\includegraphics[width = 0.30\textwidth]{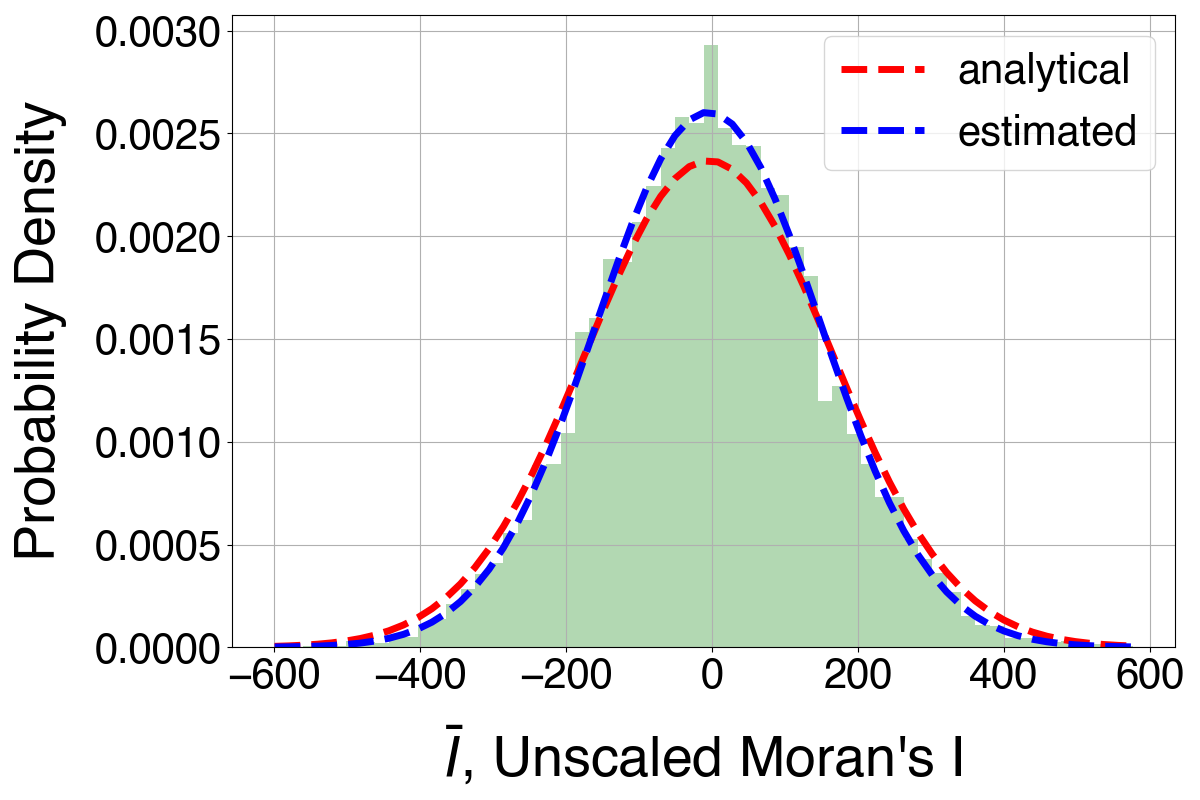}\label{fig:medium-independence}}\quad
\subfloat[low independence, $b = 25\%$. \\ $\Iuppermean = -4.39$, $\Iuppervar = 161.66$ \\ $\samplemean = -3.94$, $\samplevar = 139.85$]{\includegraphics[width = 0.30\textwidth]{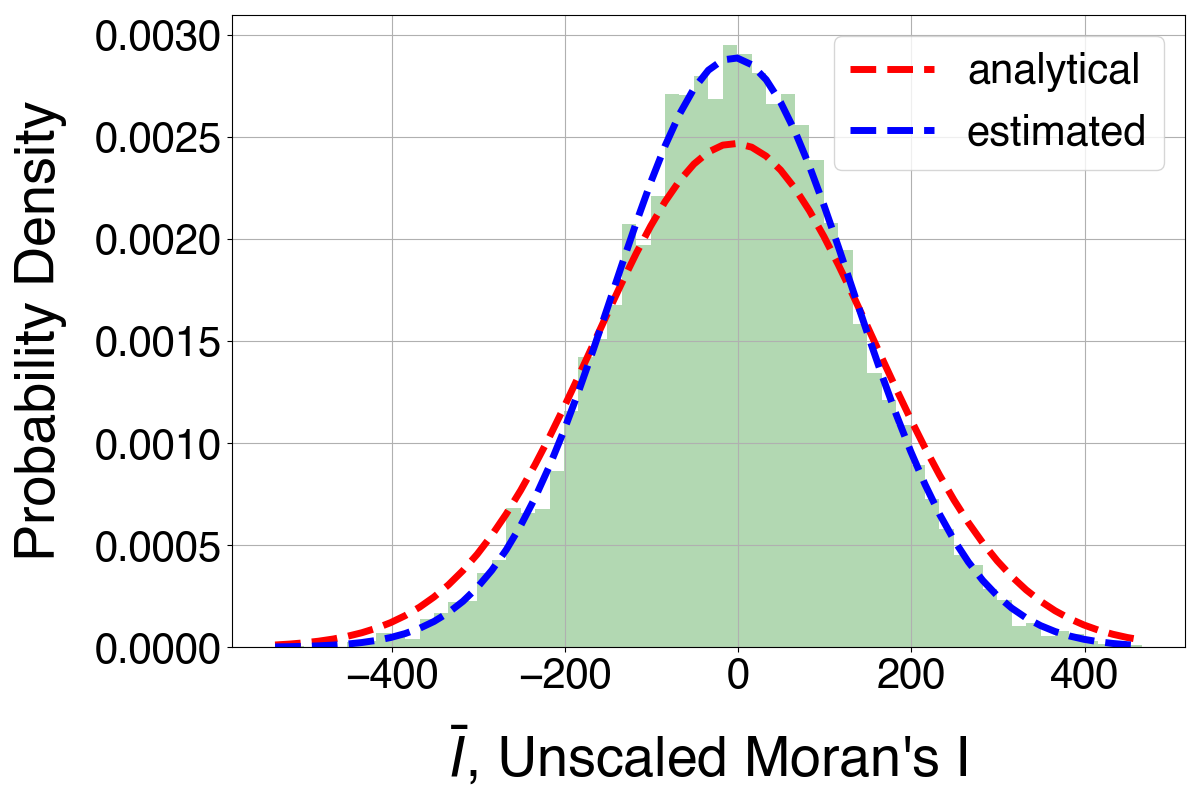}\label{fig:low-independence}}
\caption{Histograms of $\Iupper$ of 10,000 randomly generated $40\times40$ grids using rook's distance. From (a) to (c) the proportion of background $b$ decreases, i.e., the level of independence decreases. The blue lines represent the estimated normal distributions from the histograms of $\Iupper$. The red lines represent the analytical approximations based on \cref{theorem:approx-mean-theorem} and \ref{theorem:approx-variance-theorem}.} 
\label{fig:approximation-visualization}
\end{figure*}

Figure \ref{fig:approximation-visualization} demonstrates that when the independence condition is satisfied, the analytical normal approximation perfectly fits the empirical distribution. More specifically, in the high independence case, the analytical approximation passes the Kolmogorov-Smironov Test with a $p$-value of $0.20$, which is significantly larger than $0.05$. That means our approximation is statistically indistinguishable from the actual distribution. As the proportion of background decreases, the analytical standard deviation becomes increasingly overestimated, due to the violation of the independence condition. However, Figure \ref{fig:approximation-visualization} and Figure \ref{fig:independence-relaxation} also show, though the approximate variance becomes inaccurate, the approximate mean remains extremely accurate: the absolute difference between the analytical mean and the empirical mean is consistently below $2\%$ of the empirical standard deviation regardless of the level of independence. This is very useful for designing loss functions: given the mean of a normal distribution, even though we do not have its exact variance, the \textbf{relative} difference between any two values is known, i.e., we can still accurately compute the direction of the gradients.

\begin{figure*}[t!]
\centering
\captionsetup[subfloat]{justification=centering}
\subfloat[standardized difference of the mean of $\Iupper$]{\includegraphics[width = 0.31\textwidth]{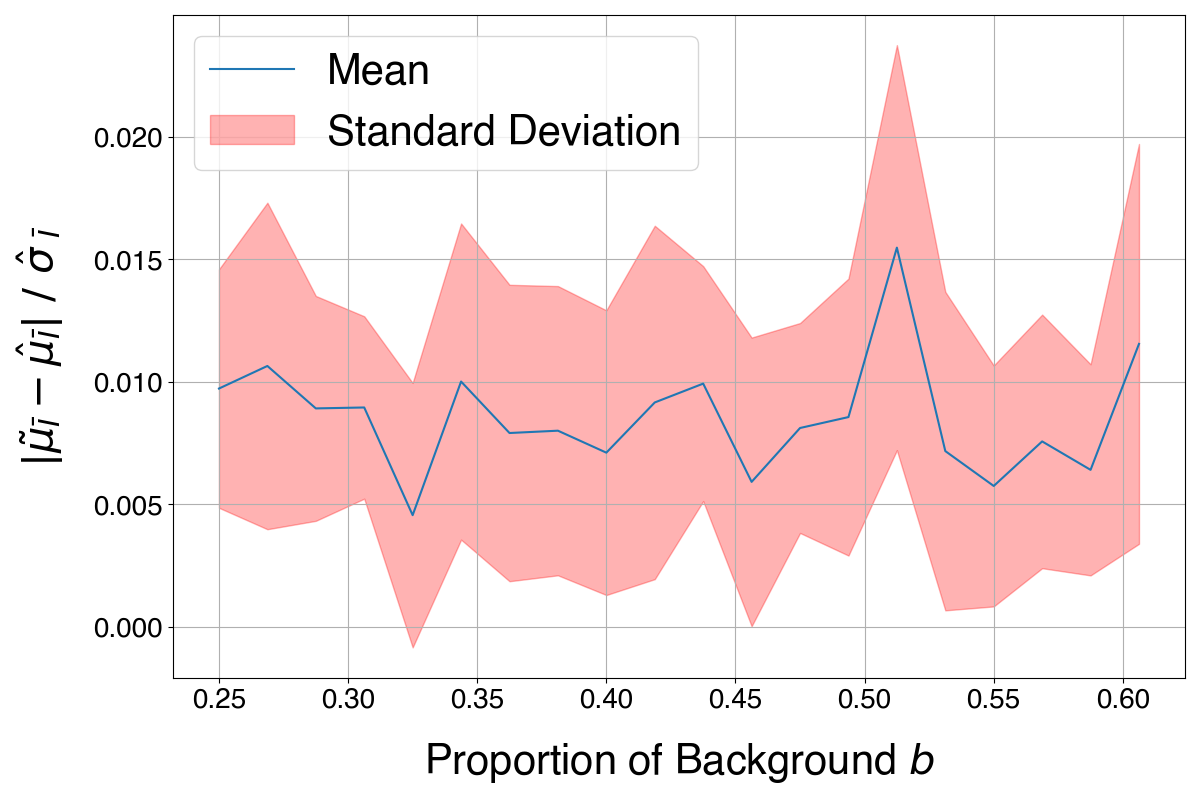}\label{fig:standardized-mean.png}}\quad
\subfloat[standardized difference of the standard deviation of $\Iupper$]{\includegraphics[width = 0.31\textwidth]{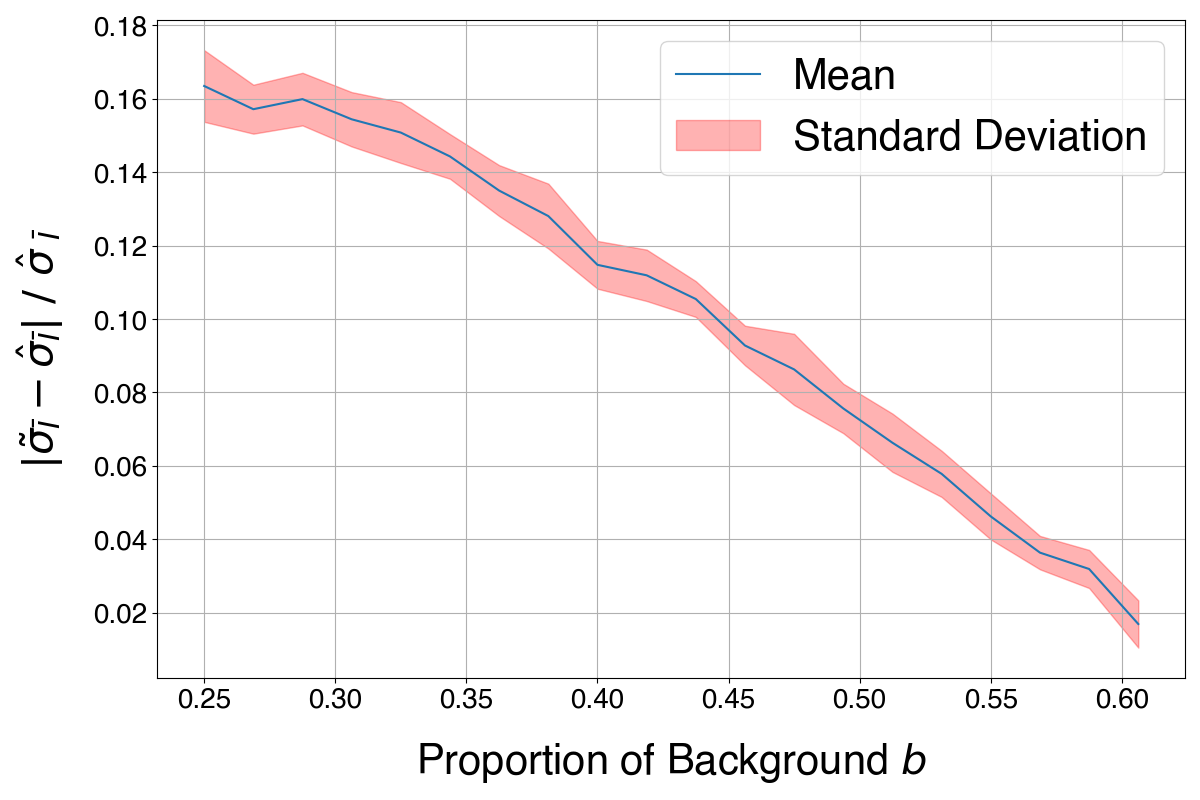}\label{fig:standardized-variance.png}}\quad
\subfloat[KL-divergence to the empirical distribution]{\includegraphics[width = 0.31\textwidth]{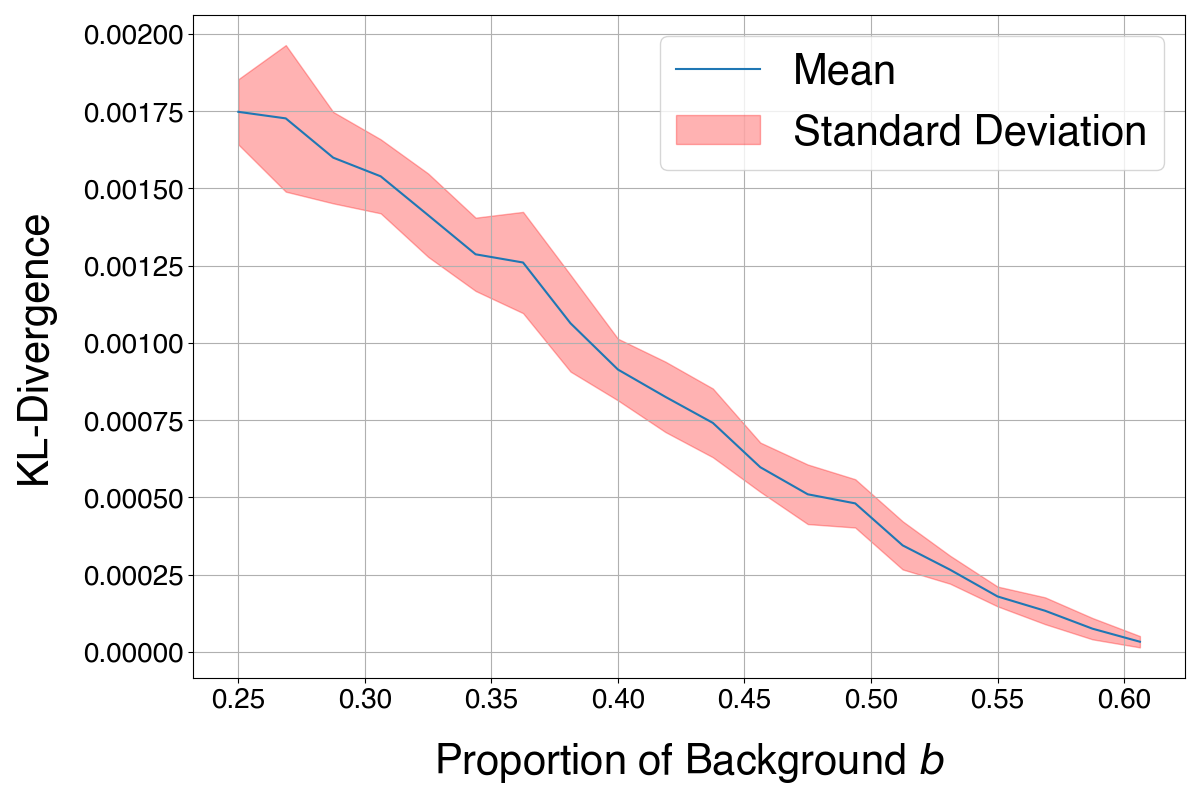}\label{fig:kl-divergence}}
\caption{The relation between the approximation accuracy and the level of independence (measured by $b= n_{r_{\max}} / N$, the proportion of background values). At each level of independence, we repeatedly sample 10,000 $40\times40$ grids randomly for 10 times. (a) and (b) plot the standardized difference between the analytical mean/standard deviation and the empirical mean/standard deviation. (c) plots the KL divergence from the analytical approximation to the empirical distribution.}
\label{fig:independence-relaxation}
\end{figure*}

\subsection{Relaxation of Assumption 3: Different Numbers of Neighbors} \label{sec:relaxation-assumption-2}

It is common that not all observations have the same number of neighbors. For example, for grid data with rook's weights, the border observations only have $3$ neighbors and corner observations only $2$, instead of $4$. There are two cases of violation: (1) the total number of neighbors remains $kN$, but not equally distributed, or (2) the total number of neighbors is larger or smaller than $kN$. 

The first case can be seen as a perturbation of the weight matrix. To investigate this, given a rate of perturbation $\rho$, we randomly select $\rho kN$ $1$s and $\rho kN$ $0$s in the weight matrix and flip their values (0 to 1 and vice versa). This is equivalent to randomly deleting and adding equal amount of edges in the graph. After such perturbation, while the total number of neighbors remains the same, the equal-number-of-neighbor assumption is violated. We generate a random sample of grids, perturb the weight matrix at increasing rates of perturbation, derive the analytical approximation according to the perturbed weight matrices, and compute the KL-divergence respectively. In Figure \ref{fig:kl-divergence-perturb-k} we plot how the approximation accuracy changes as the rate of perturbation increases. When the rate of perturbation is under 0.15, the KL-divergence remains under $2.5\times10^{-5}$, which is comparable to the high-independence KL-divergence values in Figure \ref{fig:kl-divergence}. It indicates that a mild violation of the equal-number-of-neighbor assumption shall not introduce significant approximation error.

\begin{figure*}[t!]
\centering
\captionsetup[subfloat]{justification=centering}
\subfloat[perturbation of weight matrix]{\includegraphics[width = 0.48\textwidth]{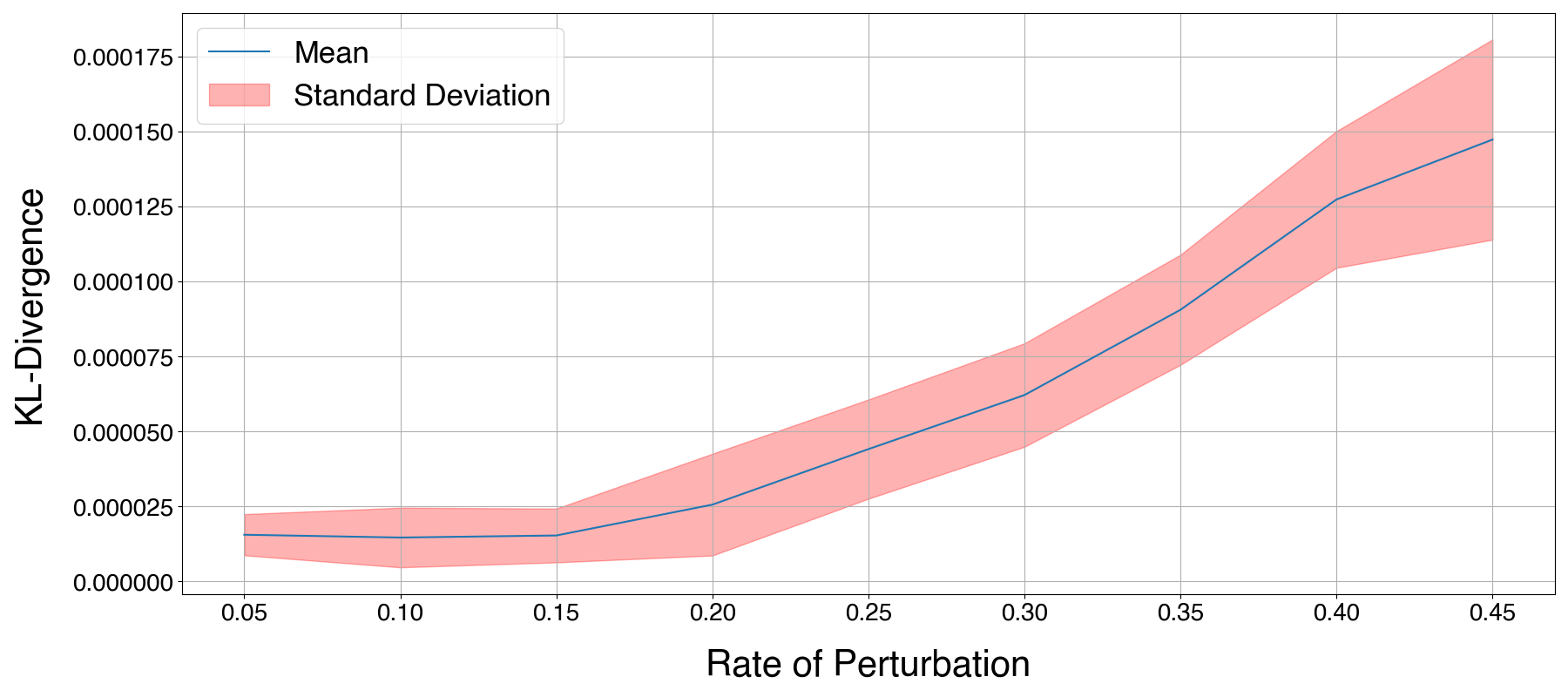}\label{fig:kl-divergence-perturb-k}}\quad
\subfloat[change of total number of neighbors]{\includegraphics[width = 0.48\textwidth]{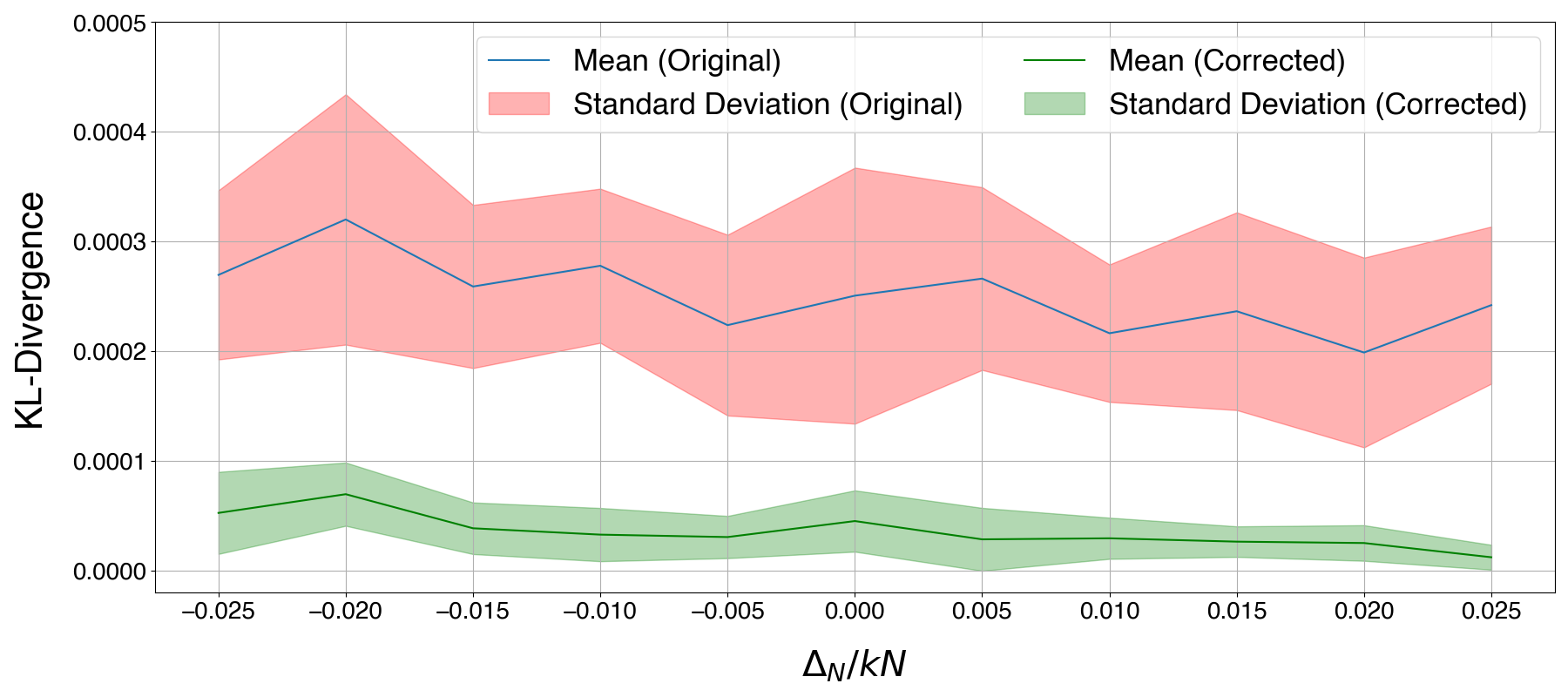}\label{fig:kl-divergence-systematic-k}}
\caption{(a) The relation between the approximation accuracy and the level of perturbation in the number of neighbors, measured by the rate of perturbation. (b) The relation between the approximation accuracy and the level of change in total number of neighbors, measured by the change rate $\Delta_N/kN$.}
\label{fig:number-of-neighbors-relaxation}
\end{figure*}

In the second case, the violation leads to systematic overestimation or underestimation of $|S_{p,q}|$. Let the difference between the total number of neighbors and $kN$ be $\Delta_N$, each $|S_{p,q}|$ can be corrected by multiplying the scaling factor $1 - \Delta_N / kN$. To verify this, we randomly select $\Delta_N$ $1$s/$0$s from the weight matrix and flip their values, which increases/decreases the total number of neighbors by $\Delta_N$. In Figure \ref{fig:kl-divergence-systematic-k} we plot how the approximation accuracy changes with $\Delta_N$ when applying and not applying the scaling factor. We can see the corrected approximation has consistently low KL-divergence ($< 1\times10^{-4})$. It means that we can safely use the approximation with spatial weights that have edge cases (like the borders and corners in the rook/queen weights).

\subsection{Relaxation of Assumption 4: Common Neighbors}

The assumption in Lemma \ref{lemma:normal-approx-same-lemma} that none of the colored vertices share common neighbors is used to derive the probability of success. When this assumption holds, each success creates exactly $2$ same-value pairs and $k-2$ candidate vertices which we can color to obtain another $2$ same-value pairs. Intuitively, it assumes that colored vertices do not form large clusters. When $n_p$ is small, the assumption is valid because the colored vertices are scattered. However, when $n_p$ is large, $|S_{p,p}|$ will be underestimated because coloring a common neighbor creates more than $2$ same-value pairs.

% \todo{A figure here to illustrate how the existence of common neighbors cause trouble.}

We find that multiplying each $|S_{p,p}|$ with a scaling factor $(|S_{p,p}| - 1) / |S_{p,p}|$ corrects the underestimation. The intuition is that, as $|S_{p,p}|$ gets large, the violation of having common neighbors has a decreasing effect on the approximate probability of success. To verify this, we generate $40\times40$ grids with rook's weights that have one background value and three foreground values. Set the value sizes of all the foreground values to be the same number $n$ from $200$ to $300$ by a step of $20$, and investigate how the analytical $|S_{p,p}|$ differs from the empirical $|\hat{S}_{p,p}|$ as $n$ increases. Figure \ref{fig:difference-beta2}  demonstrates that the corrected analytical $|S_{p,p}|$ is an extremely accurate approximation of the empirical $|\hat{S}_{p,p}|$, and \cref{fig:kl-divergence-beta2} demonstrates that such improvement results in better overall KL-divergence from the approximation to the empirical distribution. The up-going trend in Figure \ref{fig:kl-divergence-beta2} is irrelevant to the assumption of no common neighbors. Instead, it is a result of the violation of the independence condition as increasing $n$ means a smaller proportion of background, which makes the approximation of standard deviations less accurate.

\begin{figure*}[t!]
\centering
\captionsetup[subfloat]{justification=centering}
\subfloat[difference from the empirical $|\hat{S}_{p,p}|$]{\includegraphics[width = 0.48\textwidth]{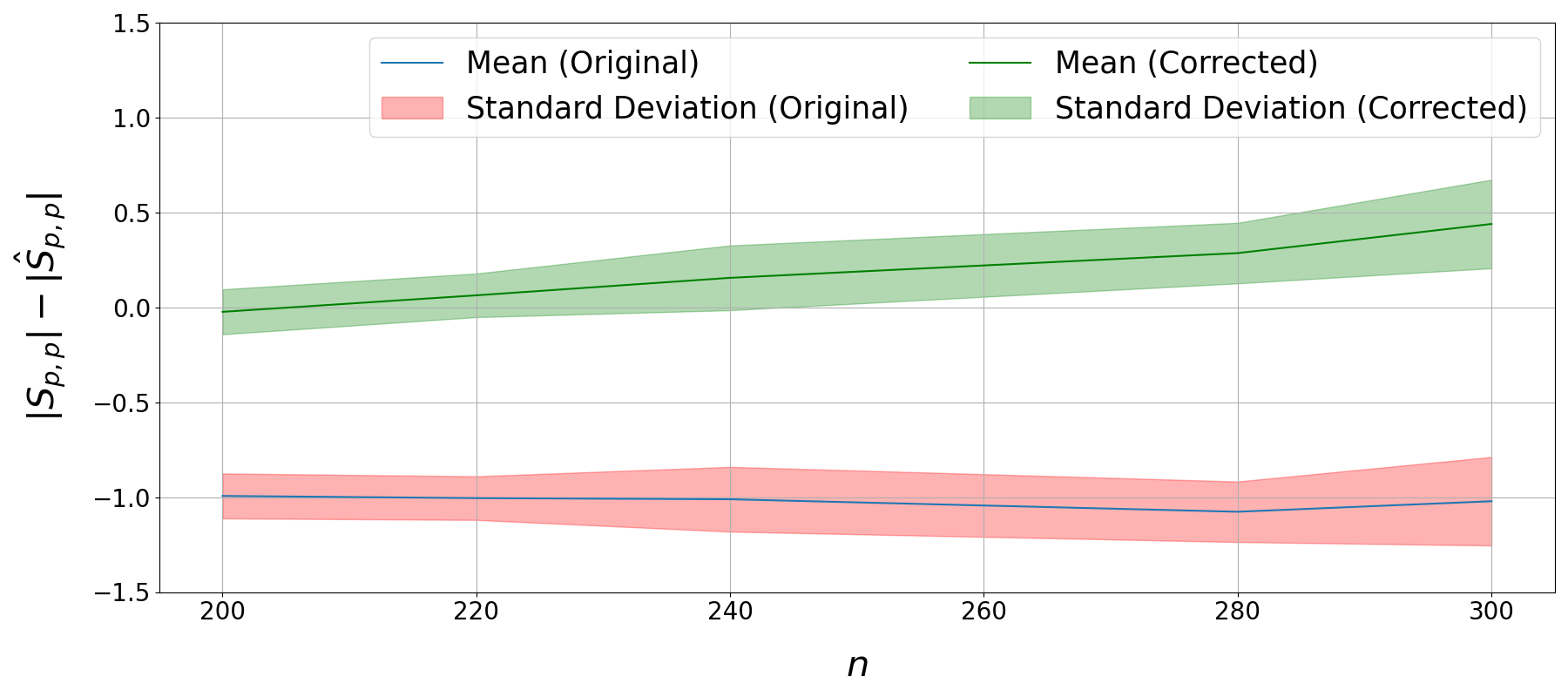}\label{fig:difference-beta2}}\quad
\subfloat[KL-divergence to the empirical distribution]{\includegraphics[width = 0.48\textwidth]{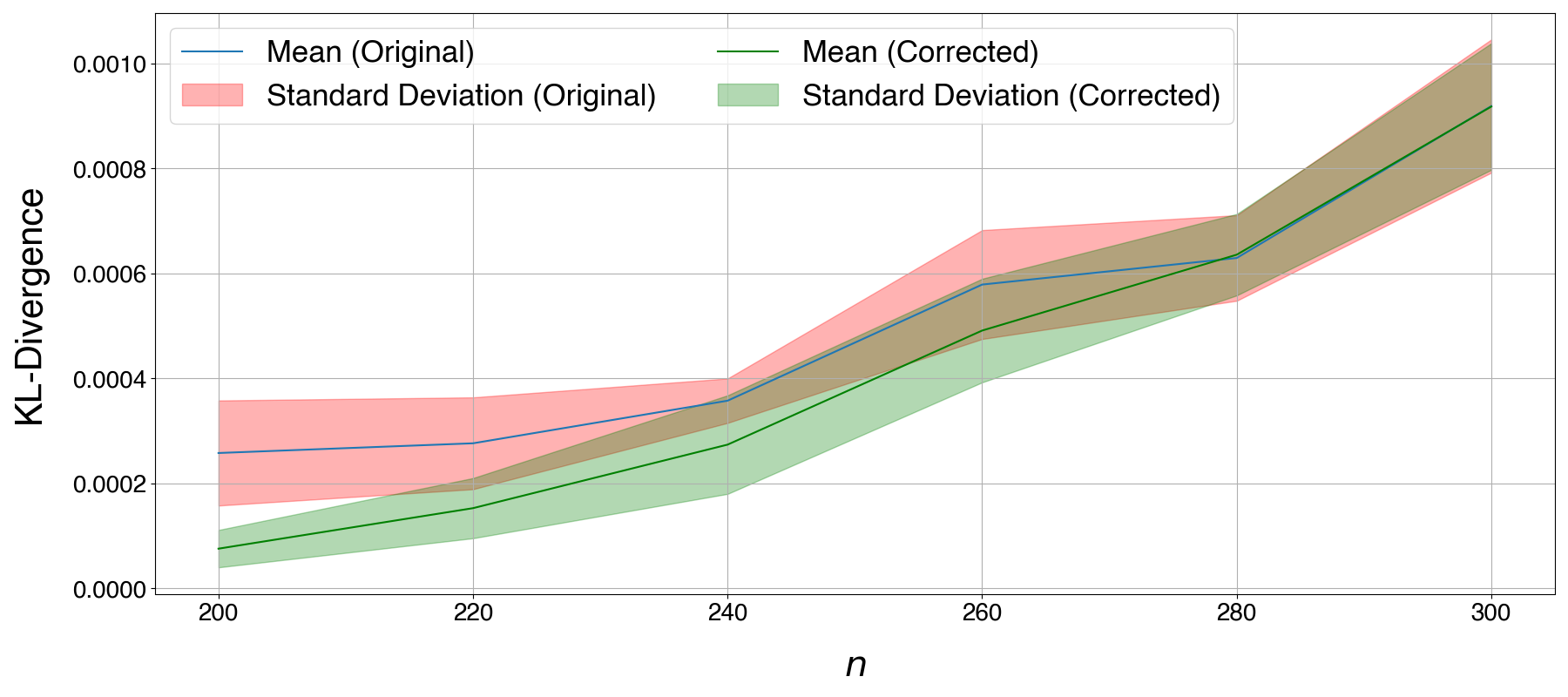}\label{fig:kl-divergence-beta2}}
\caption{(a) The uncorrected analytical $|S_{p,p}|$ is constantly underestimated, while the corrected analytical $|S_{p,p}|$ approximates the empirical $|\hat{S}_{p,p}|$ better. (b) The relation between the KL-divergence and $n_p$. The larger the $n_p$, the more common neighbors, and the worse approximation accuracy.}
\label{fig:common-neighbors-relaxation}
\end{figure*}

\subsection{Trivial Relaxations}

Some of the assumptions and conditions are trivial to satisfy. We summarize them as follows:\\

\noindent\textbf{Relaxation of Assumption 2} Data with continuous values and excessively many discrete values can be approximated by discretization (i.e., bucketization or binning). The approximation accuracy is dependent on the granularity of the buckets.\\

\noindent\textbf{Relaxation of Condition 1} This condition is automatically satisfied when the independence condition is satisfied and $k$ is not extremely large. In Lemma \ref{lemma:different-value-lemma} we use the minimum/maximum functions. That guarantees $n_q < n_p$. Then, as $k$ is usually a sufficiently large integer (e.g., the number of nearest neighbors), $n_q \ll kn_p$. Besides, when the independence condition is satisfied, $n_{r_{\max}} / N$ is sufficiently large, i.e. $N - n_{r_{\max}} \ll N$, thus $n_p < N - n_{r_{\max}} \ll N$. Then, as long as $k$ is not excessively large, $kn_p \ll N$.\\

\noindent\textbf{Relaxation of Condition 2} Violation of this condition does not significantly affect the approximation accuracy because when $n_p$ is small, its total contribution to the analytical mean and variance is also small.  

\section{Applications on Real-World Data} \label{sec:slope-experiment}

To further demonstrate the potential for practical application of our proposed method, we demonstrate how it can be applied to measure the surprisal of slope data via spatial self-information $J$.

The slope dataset\footnote{\url{https://ec.europa.eu/eurostat/web/gisco/geodata/reference-data/elevation/eu-dem/slope}} used here is obtained from the European Union statistics organization and covers the area of EU. The values in the data show slope values (i.e., $0\deg$ - $90\deg$) that are normalized into the range of 0 - 250. We split the original data into tiles of size $1000 \times 1000$, which represents a relatively homogeneous region, and further split the tiles into $50\times50$ patches. Due to the relatively small size of the patches, we bucketize the slope values with bin size 20 to avoid values that only appear very few times, i.e, merging values from $0$ to $19$ as $1$, values from $20$ to $39$ as $2$, and so forth until values from $240$ to $250$ as $13$. On the bucketized patches we compute the Moran's $I$ and the spatial self-information $J$.

\begin{figure}[ht!]
\centering
\setlength{\tabcolsep}{-1pt} % Default value: 6pt
\renewcommand{\arraystretch}{0.2} % Default value: 1
\begin{tabular}{ccccc}
  \includegraphics[width=0.20\textwidth]{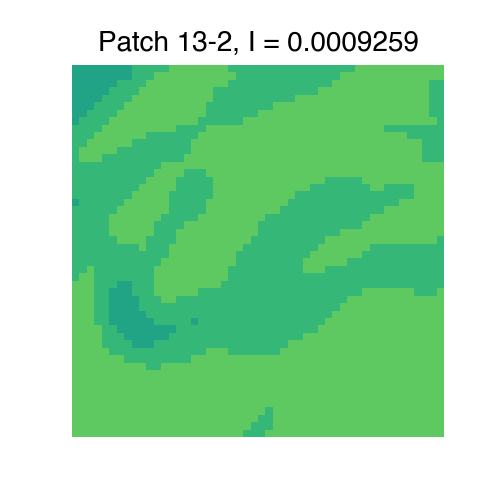} &   \includegraphics[width=0.20\textwidth]{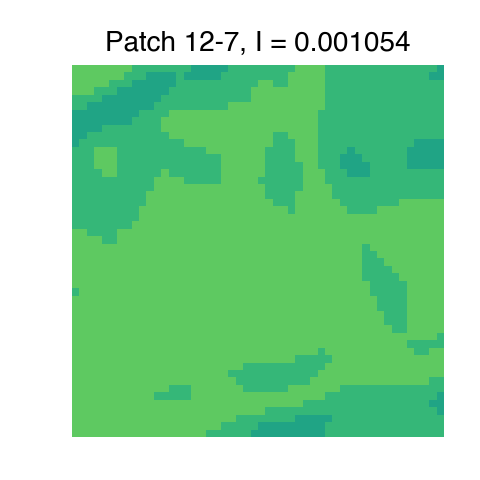} & \includegraphics[width=0.20\textwidth]{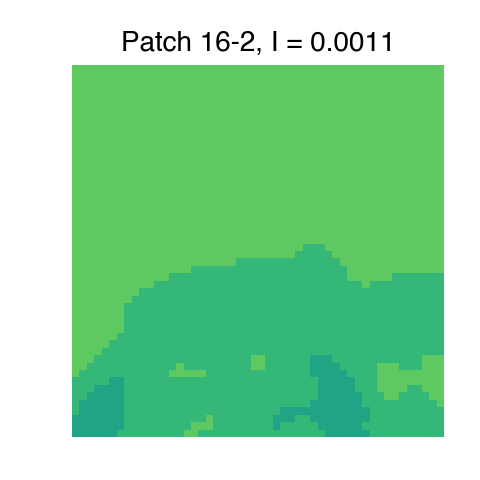} & \includegraphics[width=0.20\textwidth]{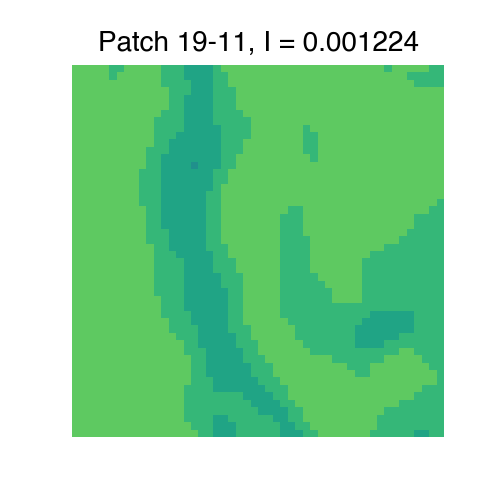} &   \includegraphics[width=0.20\textwidth]{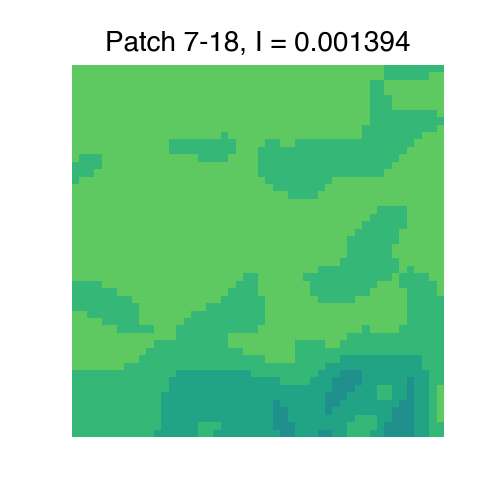}\\
% (a) first & (b) second \\[6pt]
 \includegraphics[width=0.20\textwidth]{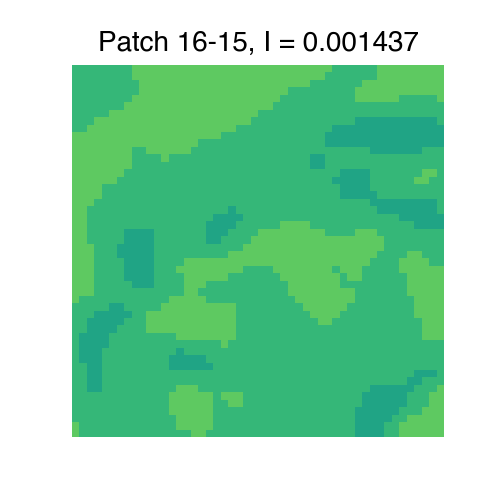} &   \includegraphics[width=0.20\textwidth]{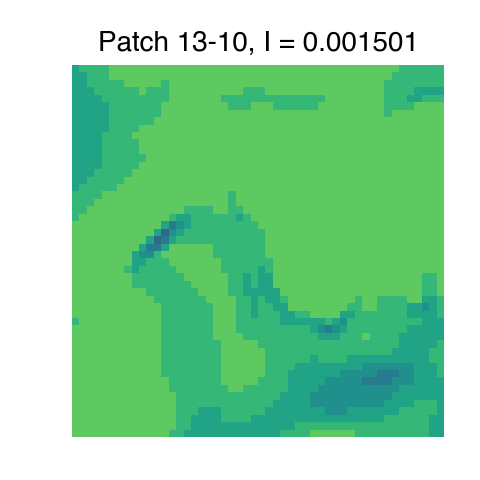} & \includegraphics[width=0.20\textwidth]{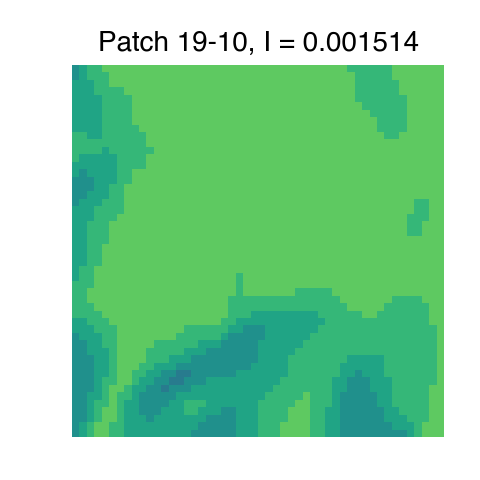} & \includegraphics[width=0.20\textwidth]{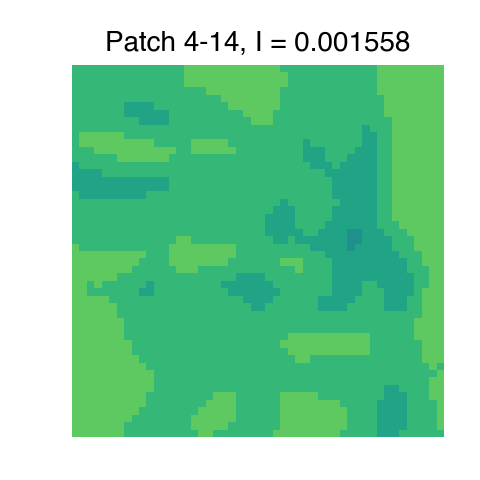} &   \includegraphics[width=0.20\textwidth]{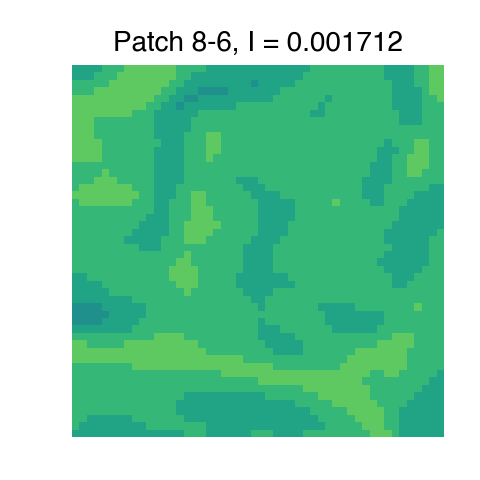}\\
% (c) third & (d) fourth \\[6pt]
 \includegraphics[width=0.20\textwidth]{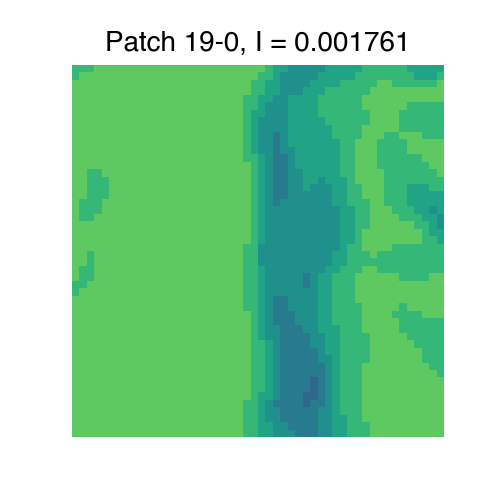} &   \includegraphics[width=0.20\textwidth]{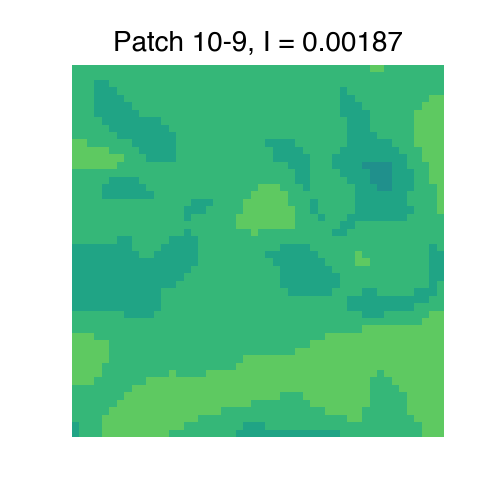} & \includegraphics[width=0.20\textwidth]{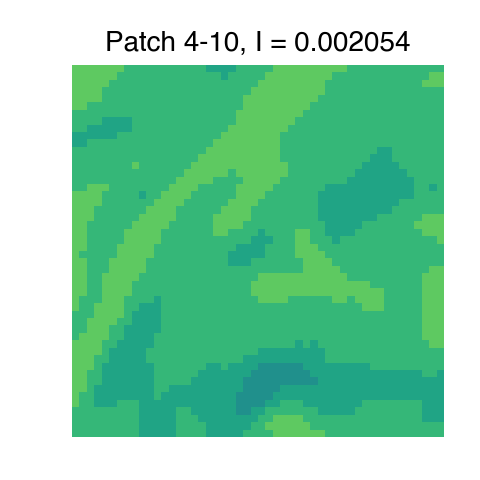} & \includegraphics[width=0.20\textwidth]{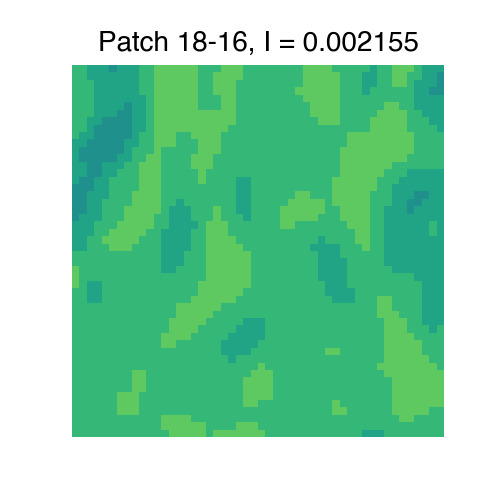} &   \includegraphics[width=0.20\textwidth]{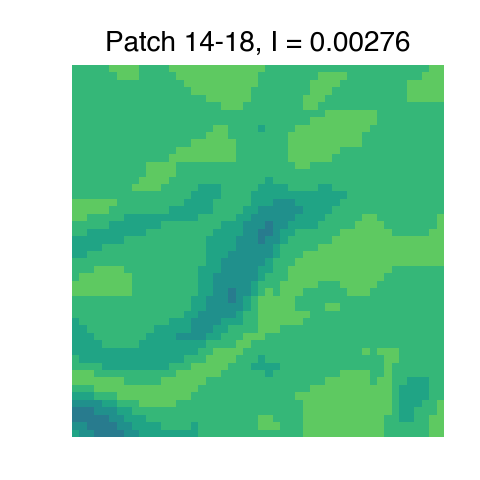}\\
% \multicolumn{2}{c}{\includegraphics[width=0.20\textwidth]{it} }\\
% \multicolumn{2}{c}{(e) fifth}
\end{tabular}
\caption{Slope patches in ascending order of Moran's I from left to right and top to bottom.}
\label{fig:moran-I-on-slope}
\end{figure}
\begin{figure}[ht!]
\centering
\setlength{\tabcolsep}{-1pt} % Default value: 6pt
\renewcommand{\arraystretch}{0.2} % Default value: 1
\begin{tabular}{ccccc}
  \includegraphics[width=0.20\textwidth]{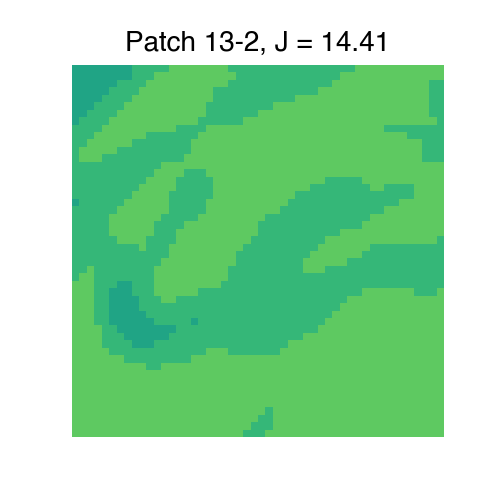} &   \includegraphics[width=0.20\textwidth]{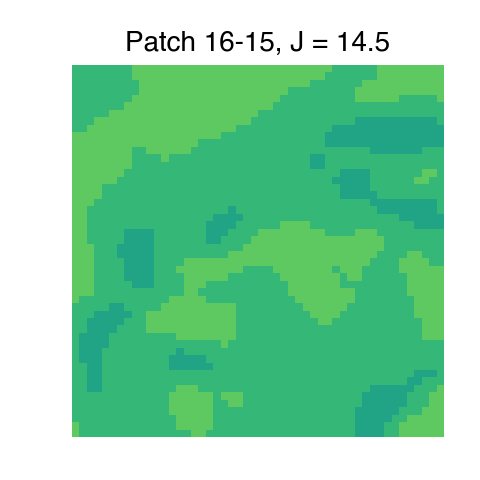} & \includegraphics[width=0.20\textwidth]{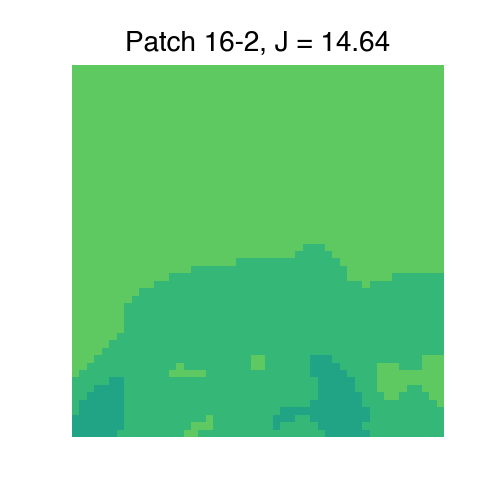} & \includegraphics[width=0.20\textwidth]{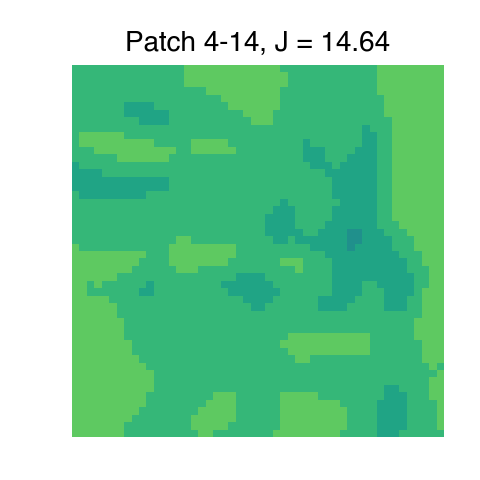} &   \includegraphics[width=0.20\textwidth]{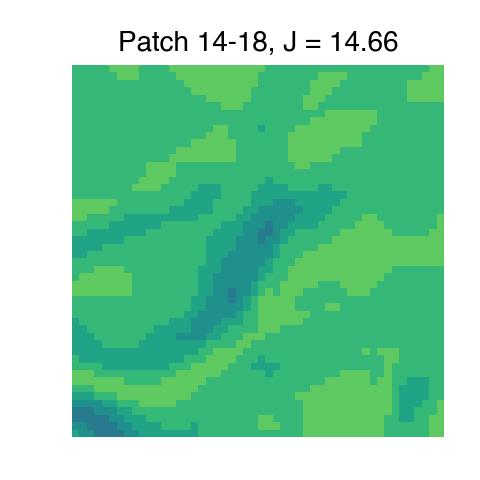}\\
% (a) first & (b) second \\[6pt]
 \includegraphics[width=0.20\textwidth]{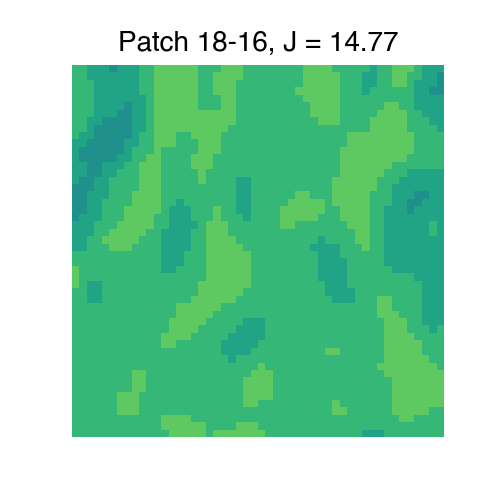} &   \includegraphics[width=0.20\textwidth]{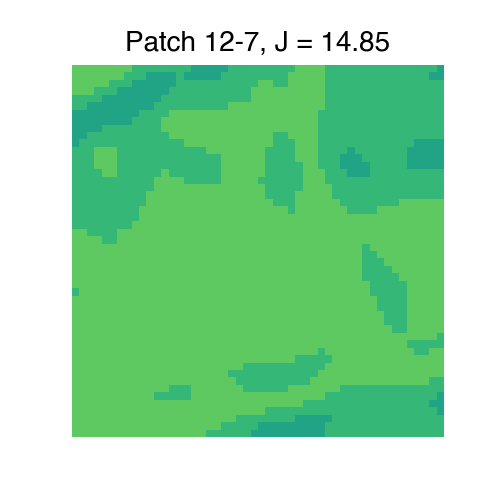} & \includegraphics[width=0.20\textwidth]{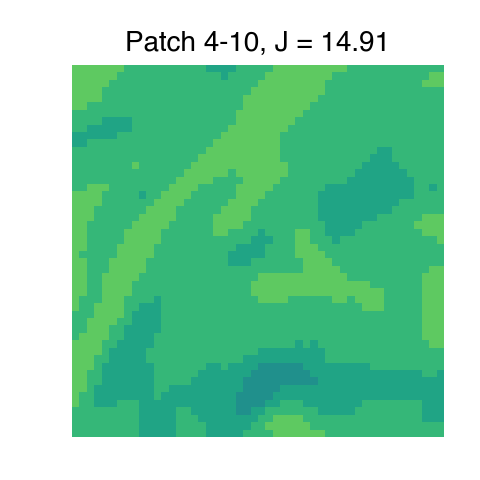} & \includegraphics[width=0.20\textwidth]{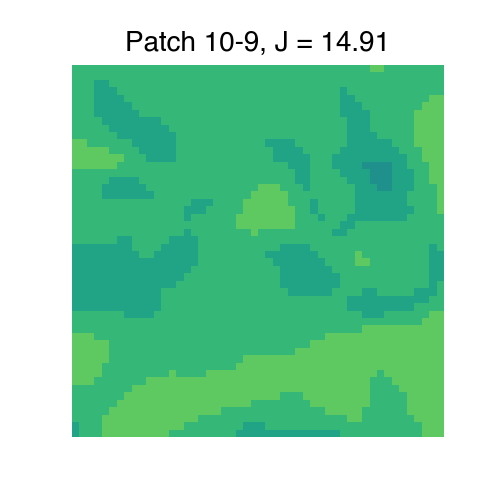} &   \includegraphics[width=0.20\textwidth]{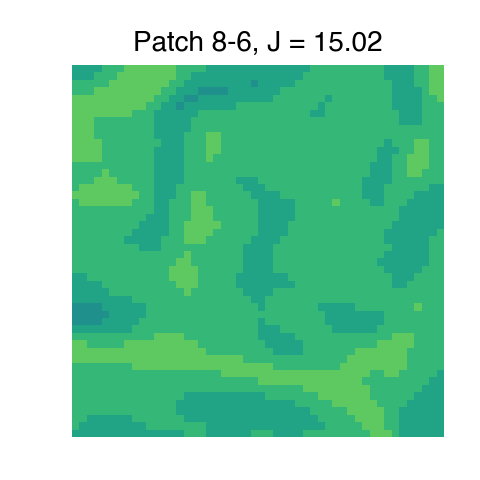}\\
% (c) third & (d) fourth \\[6pt]
 \includegraphics[width=0.20\textwidth]{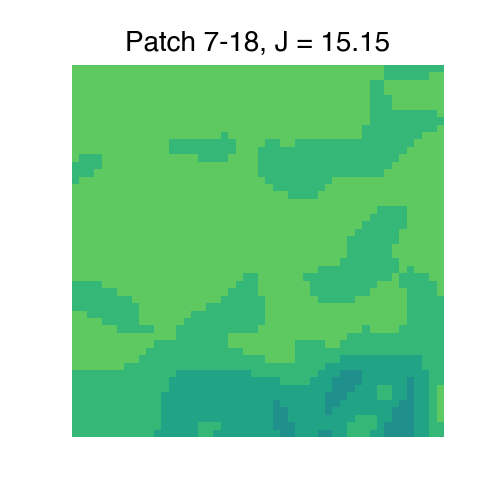} &   \includegraphics[width=0.20\textwidth]{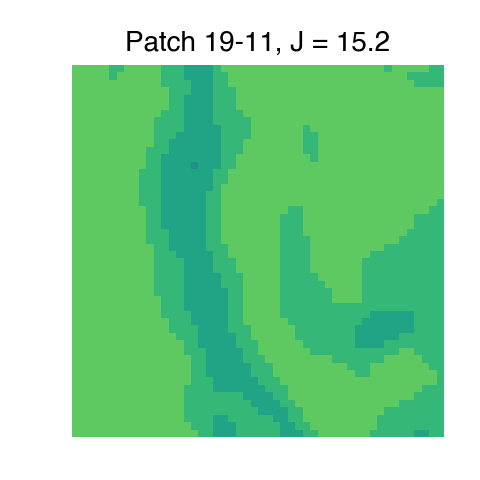} & \includegraphics[width=0.20\textwidth]{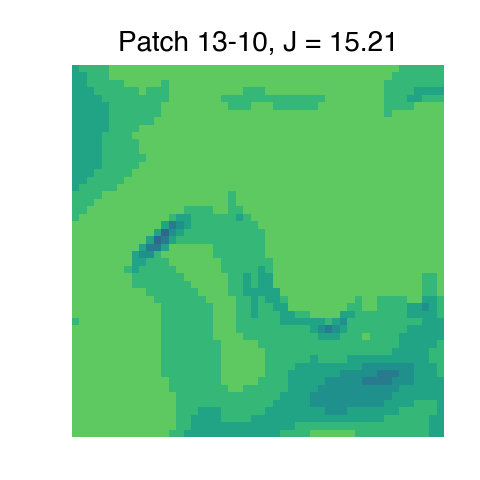} & \includegraphics[width=0.20\textwidth]{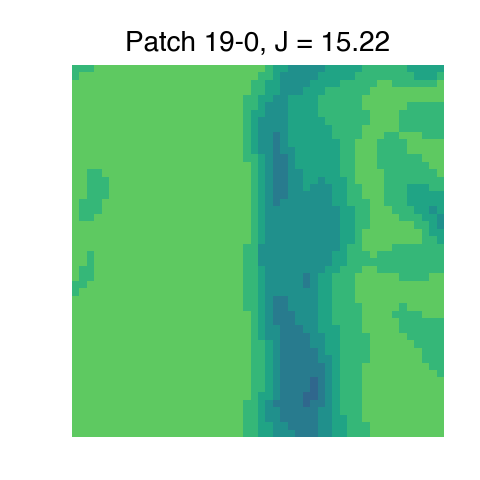} &   \includegraphics[width=0.20\textwidth]{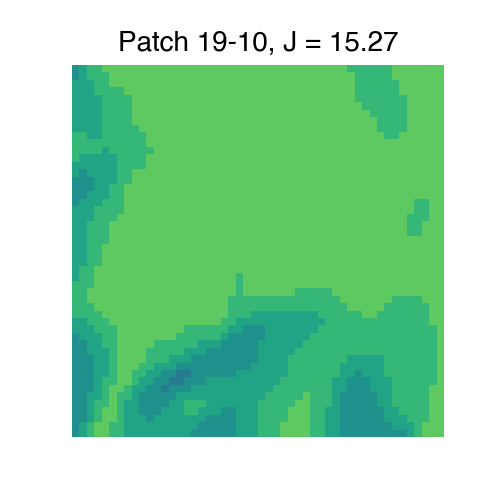}\\
% \multicolumn{2}{c}{\includegraphics[width=0.20\textwidth]{it} }\\
% \multicolumn{2}{c}{(e) fifth}
\end{tabular}
\caption{Slope patches in ascending order of surprisal from left to right and top to bottom.}
\label{fig:surprisal-on-slope}
\end{figure}

As we have pointed out, Moran's I statistics with different value schemes are not directly comparable. We only know that statistically, a Moran's I value that is significantly above $1 / (N - 1)$ or below $-1 / (N - 1)$ indicates a presence of spatial autocorrelation, which in our case equals $1 / 2500 = 0.0004$. Instead, the spatial self-information provides a unified measure of spatial autocorrelation from the perspective of information content, regardless of the value schemes as long as the independence assumption holds. To demonstrate this, we randomly select 15 patches whose proportion of background ranges from $0.55$ to $0.65$, which we know from Figure \ref{fig:high-independence} guarantees high independence. We do not require their value schemes to be identical (which is impossible). In Figure \ref{fig:moran-I-on-slope}, the patches are plotted in ascending order of Moran's I. And in Figure \ref{fig:surprisal-on-slope}, the same patches are plotted, but in ascending order of spatial self-information $J$. We can see the two orders are not the same. For example, Patch 19-0 shows a sharp spatial pattern where a strip of mountains stands against a wide plane. In Figure \ref{fig:moran-I-on-slope}, however, it has lower positive Moran's I than Patch 10-9, which is more scattered and should have weaker spatial autocorrelation. This phenomenon occurs because Patch 10-9 has larger areas of deeper green (higher values), which inflates the variance of Moran's I with its value scheme. Our spatial self-information $J$, on the contrary, is able to correctly capture the strength of spatial autocorrelation. In Figure \ref{fig:surprisal-on-slope}, Patch 19-0 has the second largest spatial self-information (15.22) compared to Patch 10-9 (14.91), meaning it is $e^{15.22 - 14.9} = 1.36$ times easier to observe Patch 10-9 than Patch 19-0, i.e., Patch 19-0 is more surprising.

\section{Conclusions and Future Work} \label{sec:conclusion}

In this paper, we theoretically derive the asymptotic analytical distribution of global Moran's I in the case of binary weights, under a series of broad randomness assumptions. We further develop a comprehensive set of techniques that efficiently computes the approximate probability and self-information of a spatial sample and corrects the error caused by violations of assumptions. Both synthetic and real-world experiments show that our approximation remains accurate and robust even if the assumptions and conditions are not ideally satisfied. Our research provides practical means to measure the information loss in spatially distributed data due to the presence of spatial autocorrelation with applications in spatial data analysis, GeoAI models, and general machine learning/deep learning. For future work, it is worth exploring to \textbf{(1)} relax the independence assumption to enable more accurate approximation on highly scattered spatial data such as maps of POIs, \textbf{(2)} derive a non-binary weight version, and \textbf{(3)} study different settings such as continuous value surfaces and continuous entropy. Finally, while our work was centered around Moran's I, similar ideas likely generalize to related concepts such as the semivariogram which could be expressed as increased entropy by distance. %Progress in these two directions will unlock even wider applications in various fields of research.

\bibliography{lipics-v2021-sample-article}

\appendix

\end{document}